\documentclass[11pt]{article}
\usepackage[left=1in,top=1in,right=1in,bottom=1in,head=0in,nofoot]{geometry}

\usepackage{setspace,url,bm,amsmath} % For double-spacing, URL font, math symbols
\usepackage{scalerel}

\usepackage{xcolor}
 \usepackage{multirow}
 \usepackage{arydshln}
 \usepackage{mathtools}
\usepackage{titlesec} % Section header formatting
\titlelabel{\thetitle.\quad} % Section header formatting
\titleformat*{\section}{\bf\Large\center}

\usepackage{authblk}
\usepackage{float}
\usepackage{graphicx} % Graphics scaling
\usepackage{bbm}
\usepackage{latexsym}
\usepackage{caption}
\usepackage[margin=20pt]{subcaption}
\usepackage{enumerate}
\graphicspath{ {./Graphs/} }

\usepackage{amsthm}
\usepackage{amssymb}
\usepackage{amsmath}
\usepackage{color}

\usepackage{comment}
\theoremstyle{definition}
\newtheorem{assumption}{Assumption}
\newtheorem*{theorem*}{Theorem}
\newtheorem{theorem}{Theorem}
\newtheorem*{rmk*}{remark}
\newtheorem{proposition}{Proposition}
\newtheorem{lemma}{Lemma}
\newtheorem{example}{Example}

\newtheorem{definition}{Definition}
\newtheorem{remark}{Remark}
\newtheorem{corollary}{Corollary}
\newtheorem*{corollary*}{Corollary}

\usepackage{natbib} % ASA citation style
\bibpunct{(}{)}{;}{a}{}{,} % ASA citation style

\usepackage{etoolbox} % Bibliography underfull/overfull box fix
\apptocmd{\sloppy}{\hbadness 10000\relax}{}{} % Bibliography underfull/overfull box fix

\usepackage{multibib}
\newcites{sec}{References}

\usepackage{color}
\usepackage{listings}
\usepackage{hyperref}
\usepackage{booktabs}

\usepackage{lscape}

\RequirePackage[normalem]{ulem}

\allowdisplaybreaks

\begin{document}

\singlespacing

\title{Evaluating Surrogates in Individualized Treatment Rules}

\author[1]{Zeyu Xu}
\author[2]{Xiaojie Mao}
\author[3]{Hao Mei}
\author[3]{Yue Liu}

\affil[1]{School of Statistics, Renmin University of China, Beijing 100872, China}
\affil[2]{School of Economics and Management, Tsinghua University, Beijing 100084, China}
\affil[3]{Center for Applied Statistics and School of Statistics, Renmin University of China, Beijing 100872, China}

%$^1$; 

%  (E-mail: staza@nus.edu.sg)

\date{}

\maketitle

\begin{abstract}
    In many decision-making problems, the primary outcome is expensive, time-consuming, or difficult to observe, so individualized treatment rules (ITRs) may be instead learned from surrogate endpoints. However, a surrogate that is highly associated with the primary outcome, or even satisfies existing surrogate criteria, may not necessarily induce a treatment rule that performs well on the primary outcome, especially under treatment resource budget constraints. In this paper, we develop a principled framework for evaluating the decision-making value of surrogate endpoints. We introduce three ITR-oriented performance measures: surrogate regret, which assesses the expected loss from using the surrogate-optimal ITR instead of outcome-optimal ITR; surrogate gain, which quantifies the benefit of surrogate-optimal ITRs relative to the no-treatment baseline; and surrogate efficiency, which evaluates improvement over random treatment assignment. We also extend them to budget-constrained settings. We propose augmented inverse probability weighted (AIPW) estimators for these measures and establish their large-sample properties. We demonstrate the proposed approach on both simulations and an application to the Criteo dataset.

% Overall, this work offers the first principled framework to rigorously evaluate surrogates, providing a promising pathway toward more effective decision-making.

% This paper evaluates the use of surrogates for policy learning, focusing on both unconstrained and constrained policy learning scenarios. We propose a surrogate measure \(\Delta V\) to characterize the difference in utilities under optimal policies learned with the primary outcome \(Y_i\) and a candidate surrogate \(S_i\). We discuss estimators for \(\Delta V\) and their theoretical properties, extending the framework to include policy class restrictions, surrogate improvement, and domain adaptation. The paper builds on existing literature and provides a comprehensive approach to surrogate evaluation in policy learning.
\end{abstract}


\newpage

\onehalfspacing

\section{Introduction}
In many decision-making problems, the primary outcome of interest is often expensive, time-consuming, or difficult to observe. Researchers and practitioners therefore rely on surrogate endpoints that can be measured more cheaply, earlier, or at larger scale. Such surrogates are widely used across domains. For example, in online advertising, advertisers typically optimize for clicks rather than conversions, so clicks are used as surrogates when making ad-serving decisions \citep{cheng2021long}. In education, test scores act as short-term surrogates of the long-run effects on earnings of early-childhood interventions \citep{athey2025surrogate}. In medicine, biomarkers such as viral load and CD4 counts in AIDS trials serve as surrogate endpoints for clinical outcomes like mortality, which often require years of follow-up \citep{fleming1994biomarkers}. When treatment decisions are individualized, a natural question is whether a treatment rule learned from a surrogate endpoint can adequately replace the rule that would have been learned from the primary outcome.

The answer is not obvious. The surrogate literature has long recognized that association between a surrogate and the primary outcome can be misleading: a treatment may improve the surrogate and yet worsen the primary outcome even if the surrogate and primary outcome are strongly associated, a phenomenon often referred to as the surrogate paradox \citep{vanderweele2013surrogate, elliott2023surrogate}. Existing surrogate criteria, however, are largely designed for causal effect evaluation rather than treatment allocation decision-making. In individualized treatment rules (ITRs), the relevant question is not simply whether the surrogate is associated with the primary outcome, but whether the treatment rule induced by the surrogate performs well when evaluated on the primary outcome.  This distinction is critical in practice. For instance, in medicine, antibody levels are often used as surrogates because they are believed to be positively correlated with clinical outcomes. However, a treatment rule that maximizes the expected surrogate endpoint may fail to improve, or even harm, the true outcome \citep{Soldatov2025}. Such discrepancies become even more pronounced when treatment resources are limited. Even if the surrogate and the primary outcome agree on the sign of treatment benefit \citep{yang2023targeting}, they may rank individuals differently, leading to substantially different treatment allocations under a budget constraint.

This paper develops a principled framework for evaluating surrogates in ITRs. Rather than asking whether a surrogate is valid for effect estimation, we ask whether it is useful for decision-making. Section \ref{sec:3} presents three counterexamples showing that high observed correlation, high potential-outcome correlation, and even sign preservation \citep{yang2023targeting} need not guarantee that a surrogate-based ITR aligns with the outcome-optimal ITR. Under budget constraints, such discrepancies can be especially severe, and a surrogate-based rule may even perform worse than random assignment.

We formalize surrogate evaluation through measures that are directly tied to decision quality. In the unconstrained setting, we define \emph{surrogate regret}, the expected loss in the primary outcome from using the surrogate-optimal ITR rather than the outcome-optimal ITR. In budget-constrained settings, where at most a fraction $\lambda$ of the population can be treated \citep{bhattacharya2012inferring, matrajt2021vaccine, xu2024optimal}, we define \emph{$\lambda$-surrogate regret}, \emph{$\lambda$-surrogate gain}, and \emph{$\lambda$-surrogate efficiency}. These measures quantify, respectively, the loss relative to the budget-optimal outcome rule, the gain relative to no treatment, and the improvement over random assignment. Together, they provide a direct assessment of the decision value of a surrogate.

We further develop an estimation and inference framework for these measures. A notable feature of our setting is that the primary outcome and the surrogate need not be observed on the same units: our methods allow one dataset containing $(A,X,Y)$ and another containing $(A,X,S)$. This two-sample design is practically relevant when long-term outcomes are available only in a smaller or slower study, whereas surrogate information can be collected more broadly or more quickly in another. We construct augmented inverse probability weighted (AIPW) estimators for the proposed measures and establish their large-sample properties. Our framework can be also extended to the single-sample case.

Empirically, we study the proposed framework in simulations and in an application to the Criteo dataset, where conversion is the primary outcome and exposure and visit are candidate surrogates. The proposed measures provide a direct way to compare candidate surrogates by the quality of the treatment rules they induce, rather than by association-based criteria alone.

The remainder of the paper is organized as follows. Section 2 introduces notation and assumptions. Section 3 presents motivating counterexamples. Section 4 defines the proposed surrogate evaluation measures. Section 5 develops estimators and establishes their large-sample properties. Section 6 reports simulation studies and the Criteo application. Section 7 concludes with extensions and limitations.

\section{Setup}

% % We consider the problem of learning an optimal treatment policy when the true outcome of interest is unobserved or costly to obtain, but one or more surrogate variables are readily available. review
% !!Let the observed data for a unit be a tuple $O = (X, A, S, Y)$ drawn from a distribution $P$. Here, $X \in \mathcal{X} \subseteq \mathbb{R}^{d_x}$ represents a vector of pre-treatment covariates, $A \in \{0, 1\}$ is a binary treatment assignment, $S \in \mathcal{S} \subseteq \mathbb{R}^{d_s}$ is a vector of surrogate variables, and $Y \in \mathcal{Y} \subseteq \mathbb{R}$ is the true scalar outcome of interest.

% We adopt the potential outcomes framework~\citep{neyman1923potential, rubin2005potential} to formalize causal effects. For each unit, we define the pair of potential outcomes $(Y(1), Y(0))$, representing the outcomes that would be observed under treatment and control, respectively. Similarly, since the surrogate variables can also be affected by the treatment, we define their potential outcomes as $(S(1), S(0))$.

% The following section demonstrates through concrete numerical examples how each of the three traditional assumptions can fail, motivating the need for our proposed evaluation framework.

\subsection{Basic Notation and Assumptions}
For each unit $i$, let $A_i \in \{0,1\}$ denote the treatment assignment ($A_i = 1$ indicates treatment, $A_i = 0$ indicates control), $X_i \in \mathcal X \subseteq \mathbb R^d$ denote baseline covariates measured prior to treatment (e.g., age, gender), $Y_i \in \mathbb{R}$ denote the primary outcome, and $S_i \in \mathbb{R}$ denote the surrogate. Following the potential outcomes framework \citep{neyman1990application, rubin2005potential}, $Y_i(a) \in \mathbb{R}$ represents the potential outcome value for unit $i$ under treatment $a \in \{0, 1\}$, and $S_i(a) \in \mathbb{R}$ represents the corresponding potential surrogate value. As each unit receives either treatment or control, we have $Y_i = \left(1 - A_i\right)Y_i(0) + A_i\,Y_i(1)$ and $S_i = \left(1 - A_i\right)S_i(0) + A_i\,S_i(1)$.

For unit $i$, the individual treatment effect on $Y$ is defined as $\tau_{Y,i} = Y_i(1) - Y_i(0)$. Similarly, the effect on $S$ is defined as $\tau_{S,i} = S_i(1) - S_i(0)$.
The conditional average treatment effect (CATE) on $Y$ is $\tau_Y(x) = E\left[\tau_{Y,i} \mid X_i = x\right] = E\left[Y_i(1) - Y_i(0) \mid X_i = x\right]$, which represents the difference in the conditional expectation of the potential outcomes between receiving and not receiving the treatment. Similarly, the CATE on $S$ is defined as $\tau_S(x) = E\left[\tau_{S,i} \mid X_i = x\right] = E\left[S_i(1) - S_i(0) \mid X_i = x\right]$.

We consider two \emph{disjoint} and \emph{independent} datasets\footnote{{ When only a single dataset $\mathcal{D}= \{(A_i, X_i, S_i,Y_i) \}_{i=1}^m$ is available, our proposed framework can be also extended. We refer to Supplementary Material Section \ref{appendix: B2} for the detailed estimation procedure.}}: the first dataset (referred to as the \emph{outcome dataset}) has observations $\mathcal{D}_1 = \{(A_i, X_i, Y_i): i \in I_1\}$ with sample size $n_1 = |I_1|$, and the second dataset (referred to as the \emph{surrogate dataset}) has observations $\mathcal{D}_2 = \{(A_i, X_i, S_i): i \in I_2\}$ with sample size $n_2 = |I_2|$. We consider the asymptotic regime where $\rho = \lim_{n_1, n_2 \to \infty} \frac{n_1}{n_2} \in (0,\infty)$. We assume that $\{A_i, X_i, S_i(0), S_i(1), Y_i(0), Y_i(1)\}$ for all $i \in I_1 \cup I_2$ are independently and identically distributed from a population distribution ${P}$. 
Throughout, we use unsubscripted variables to denote generic observations drawn from the population distribution.
% In addition, when only a single dataset $\mathcal{D}= \{(A_i, X_i, S_i,Y_i) \}_{i=1}^m$ is available, {\color{blue}we partition the index set $\{1, \ldots, m\}$ into two disjoint sets $I_1$ and $I_2$ such that $I_1 \cup I_2 = \{1, \ldots, m\}$, still defining $\mathcal{D}_1 = \{(A_i, X_i, S_i, Y_i): i \in I_1\}$ and $\mathcal{D}_2 = \{(A_i, X_i, S_i, Y_i): i \in I_2\}$. Since the observations in $\mathcal{D}_1$ and $\mathcal{D}_2$ are disjoint, they remain independent;} see Supplementary Material Section B.2 for more for details. Throughout, we use unsubscripted variables to denote generic observations from the population.

% we randomly split it into two subsets in a fixed ratio $\nu$,

We impose the following standard causal identification assumptions.

\begin{assumption}[Unconfoundedness]
\label{ass:1}
\(
    A \perp\!\!\!\perp Y(a) \mid X\), 
    \(A \perp\!\!\!\perp S(a) \mid X, \quad \text{for } a \in \{0, 1\}.\)
\end{assumption}

\begin{assumption}[Overlap]
\label{ass:2}
There exists a constant $\epsilon \in (0, 1)$ such that the propensity score $e(X) = P(A = 1 | X)$ satisfies $e(X) \in [\epsilon, 1 - \epsilon]$ almost surely. 
\end{assumption}

The unconfoundedness assumption ensures that treatment assignment is as good as random after conditioning on observed covariates $X$ \citep{imbens2015causal}.
% , which are assumed to capture all confounding variables that affect both the treatment and the outcome  
Together with the overlap assumption (which states that every unit has a positive probability of receiving either treatment or control given their covariates), this assumption enables identification of causal effects from observational data. 

We define an ITR as a deterministic map $\pi: {\mathcal{X}} \to \{0,1\}$ from the covariate space to treatment assignment. 
Let $E\left[Y(\pi(X))\right]$ and $E\left[S(\pi(X))\right]$ denote the expected outcomes under an ITR $\pi$ with respect to $Y$ and $S$, respectively.
When there is no treatment allocation budget constraint, the ITR $\pi_Y(X) = \mathbf{1}\left(\tau_Y(X) > 0\right)$ is an outcome-optimal ITR that maximizes the expectation of primary outcome $E\left[Y(\pi(X))\right]$ \citep{Kitagawa2018}. Similarly, $\pi_S(X) = \mathbf{1}\left(\tau_S(X) > 0\right)$ is a surrogate-optimal ITR that maximizes  the expected surrogate $E\left[S(\pi(X))\right]$.

% one of the optimal ITRs that maximizes the value for $Y$ is $\pi_Y = \mathbf{1}\left(\tau_Y(X) > 0\right)$, which depends on the CATE on $Y$.

\subsection{Related Literature}

A large body of literature studies surrogate endpoints in the context of causal effect evaluation. To ensure that the treatment effect on a surrogate can reliably capture the treatment effect on the primary outcome, a variety of surrogate criteria have been proposed. The first formalized criterion is the \emph{statistical surrogacy criterion} of \citet{prentice1989surrogate}, which requires the primary outcome to be conditionally independent of the treatment given the surrogate. Subsequent work has developed alternative criteria to address limitations of early frameworks, including the \emph{principal surrogate criterion} \citep{frangakis2002principal}, \emph{strong surrogate criterion} \citep{lauritzen2004discussion}, and \emph{consistent surrogate criterion} \citep{gengzhi2007}, among others. For a comprehensive overview, see \citet{vanderweele2013surrogate}. 

% To ensure that the treatment effect on a surrogate can reliably capture the treatment effect on the primary outcome, a variety of surrogate criteria have been proposed. The first formalized criterion is the \emph{statistical surrogacy criterion} put forward by \citet{prentice1989surrogate}, which mandates that the primary outcome be conditionally independent of the treatment when the surrogate is held constant. In the decades since, numerous alternative criteria have emerged to address limitations of early frameworks, including the \emph{principal surrogate criterion} \citep{frangakis2002principal}, \emph{strong surrogate criterion} \citep{lauritzen2004discussion}, and \emph{consistent surrogate criterion} \citep{gengzhi2007}, among others. For a thorough overview of existing surrogate criteria, readers are referred to \citet{vanderweele2013surrogate}.

% These criteria make important contributions to understanding surrogates. However, they are primarily designed to assess whether a surrogate can substitute for the primary outcome in causal inference, rather than to assess the quality in individualized treatment decision-making. Our work takes a complementary perspective and focuses on evaluating the decision value of surrogates. 

More recent work focuses on quantifying the extent to which a surrogate explains treatment effects. For example, \citet{wang2020model} propose a model-free optimal transformation framework \footnote{{ In Supplementary Material Section \ref{appendix 1}, we use an example to illustrate the differences between our framework and the optimal transformation framework.}} to quantify the proportion of treatment effect explained by a surrogate. \citet{agniel2024robust} develop a method to define and estimate the proportion of the treatment effect explained by a longitudinal surrogate marker, when the primary outcome is censored. Building on \citet{wang2020model}, \citet{wang2023robust} propose to optimally combine multiple markers and enhance surrogacy. These approaches provide principled tools for surrogate evaluation in causal effect estimation.

Our work differs from this line of research in that we focus on surrogate evaluation for individualized treatment rules (ITRs), where the goal is to assess the quality of treatment decisions rather than the validity of effect substitution. 

Another line of related literature studies ITRs, including both policy evaluation and the derivation of optimal treatment rules. Under unconfoundedness, a variety of approaches have been developed for learning optimal ITRs, including regression-based methods, classification-based methods, and approaches based on inverse probability weighting and doubly robust estimation \citep[e.g.,][]{qian2011performance, zhao2012estimating, athey2021policy}. Extensions have been developed to address more complex settings, such as unmeasured confounding using instrumental variables \citep{xu2023instrumental, cui2021semiparametric, chen2023estimating} or proxy variables \citep{shen2023optimal, qi2024proximal}. More recently, \citet{yang2023targeting} propose to use both covariates and surrogates to impute missing long-term outcomes, subsequently approximating the optimal targeting ITR on the imputed outcomes. Additionally, \citet{wu2024policy} develop a framework for learning optimal ITRs that balance long-term and short-term rewards.

% Within the ITR framework, prior studies have extensively investigated the evaluation \citep{manski2004statistical, chen2024robust} and derivation of optimal ITRs \citep{huang2019multicategory, liu2018augmented, ma2023learning, wang2018learning}. 
% Various methodologies have been proposed to derive optimal ITRs, ranging from penalized least squares \citep{qian2011performance} and support vector machines \citep{zhao2012estimating} to approaches relying on inverse probability weighting and doubly-robust methods for identification and estimation \citep{athey2021policy}. Another line of literature leverages instrumental variables to address unmeasured confounding \citep{xu2023instrumental, cui2021semiparametric, chen2023estimating}. More recently, \citet{yang2023targeting} propose a doubly-robust approach that uses both covariates and surrogates to impute missing long-term outcomes, subsequently approximating the optimal targeting ITR on the imputed outcomes. Additionally, \citet{wu2024policy} develop a framework for learning optimal ITRs that balance long-term and short-term rewards.

While this literature focuses on learning or evaluating treatment rules, our goal is different: we aim to evaluate the value of a surrogate itself for decision-making. We develop evaluation measures that assess surrogates through the quality of the treatment rules they induce on the primary outcome, and we also study their estimation and inference.

% While these studies focus on learning the optimal ITR, we aim to evaluate the utility of the surrogate itself in decision-making. We propose three novel evaluation measures. To our knowledge, we are the first to formulate this problem, providing new theoretical and practical insights for surrogate evaluation.

\section{A New ``Surrogate Paradox'' in ITRs}
\label{sec:3}

We now study the following question: can surrogates be used to learn individualized treatment rules (ITRs) when the ultimate goal is to maximize the expectation of the primary outcome? The key issue is not simply whether the surrogate endpoint is associated with the primary outcome, but whether the treatment rule induced by the surrogate performs well when evaluated on the primary outcome itself.

This question arises naturally in many applications. For example, in online advertising, advertisers place ads through publisher platforms. Ideally, pricing would be based on realized conversions, which directly quantify returns. However, conversions are often delayed and therefore unavailable at the time of pricing. Publishers instead rely on surrogate endpoints, such as exposure and visit rates, to make pricing or allocation decisions. This creates a potential misalignment. Advertisers seek to maximize actual conversions, whereas publishers may design ITRs that optimize the expectation of these short-term surrogates. Whether such surrogate-optimal decisions also lead to good long-term outcomes is not obvious.

This section examines several intuitive and commonly used notions of surrogate quality and shows that they are insufficient for ensuring good individualized treatment decisions. These notions include observed association between the surrogate and the primary outcome, dependence between their potential outcomes, and agreement in the sign of conditional average treatment effects. While these properties may appear desirable and can be used as heuristic indicators of surrogate quality, we show through concrete examples that none of these properties is sufficient to guarantee good ITRs, especially under budget constraints.

\begin{example}[Observed Outcome Correlation]
\label{ex:simple_correlation}
Let $A$ and $X$ be independent $\text{Bernoulli}(0.5)$ variables, and let $\varepsilon_{S}, \varepsilon_{Y} \sim N(0, 1)$ independently. We specify the structural equation model as:
$$ S = A + \alpha X + \varepsilon_{S}, \quad Y = -A + \alpha X + \varepsilon_{Y}. $$
Under this specification, the CATEs are constant across the population: $\tau_S(x) \equiv 1$ and $\tau_Y(x) \equiv -1$. The correlation coefficient between the observed surrogate $S$ and observed outcome $Y$ is:
$$ \rho(S,Y) = \frac{\alpha^2 - 1}{\alpha^2 + 5}. $$

As $\alpha$ tends to infinity, the correlation coefficient between $S$ and $Y$ approaches 1. However, the treatment effects have completely opposite signs. Since $\tau_S(x)$ is universally positive, the surrogate-optimal ITR would recommend treatment for all units. Conversely, since $\tau_Y(x)$ is universally negative, the outcome-optimal ITR would recommend treatment for none. Consequently, high observed correlation does not justify using the surrogate endpoint $S$ to construct an ITR aimed at maximizing the expectation of the primary outcome $Y$; the resulting ITR may even be uniformly harmful to the primary outcome.
\end{example}
% % This counterintuitive high correlation stems from the dominant effect of covariate $X$ on both outcomes, which overshadows the treatment effects.
% To address the fact that simple correlation can be confounded by covariate effects, we adopt a structural model specified to capture a positive correlation between potential outcomes.
\begin{example}[Potential Outcome Correlation]
\label{ex:potential_correlation}
Let $A$ and $X$ be independent $\text{Bernoulli}(0.5)$ variables, and let $\varepsilon_{S}, \varepsilon_{Y} \sim N(0, 1)$ independently. We specify the structural equation model as:
\begin{align}
S &= \beta X + \left(1-A\right) \alpha + \varepsilon_{S}, \notag\\
Y &= \beta X - \left(1-A\right) \alpha + \varepsilon_{Y}, \notag
\end{align}
where the potential outcomes are given by:
\begin{align}
S(1) &=  \beta X + \varepsilon_{S}, \quad S(0) = \alpha + \beta X + \varepsilon_{S}, \notag\\
Y(1) &= \beta X + \varepsilon_{Y}, \quad Y(0) = -\alpha + \beta X + \varepsilon_{Y}. \notag
\end{align}
The correlation coefficients between potential outcomes are:
$$\rho(S(1), Y(1)) = \rho(S(0), Y(0)) = \frac{\beta^2}{\beta^2 + 4}.$$

As $\beta$ tends to infinity, the correlation coefficients between potential outcomes $(S(a), Y(a))$ approach 1 for $a \in \{0, 1\}$, driven by the dominant common term $\beta X$. However, the treatment effects remain strictly opposite: $\tau_S(x) \equiv -\alpha$ and $\tau_Y(x) \equiv \alpha$. The surrogate-optimal ITR assigns no treatment, while the outcome-optimal ITR assigns universal treatment. This example demonstrates that even a strong potential-outcome relationship between the surrogate and the primary outcome does not guarantee correct treatment decisions.
% As $\beta$ increases, the correlations between $(S(1), Y(1))$ and between $(S(0), Y(0))$ both approach 1, since these pairs become dominated by the common term $\beta X$. This creates the appearance of a highly informative surrogate with strong predictive power for the primary outcome.
% However, the treatment effects exhibit fundamentally opposite patterns: for all units, $\tau_S(x) = S(1) - S(0) \equiv -\alpha$ and $\tau_Y(x) = Y(1) - Y(0) \equiv \alpha$. When $\alpha \neq 0$, these treatment effects have opposite signs, leading to completely contradictory ITR recommendations despite the high correlations between potential outcomes.
\end{example}

% This failure reveals that high potential outcome correlation coefficient cannot guarantee that $\tau_S(x)$ and $\tau_Y(x)$ have the same sign. 
A more promising candidate is sign preservation \citep{yang2023targeting}, which ensures that the CATEs $\tau_S(x)$ and $\tau_Y(x)$ share the same sign. In unconstrained settings, this condition guarantees that the surrogate-optimal ITR is identical to the outcome-optimal ITR. However, under budget constraints, sign preservation is still insufficient.
% because optimal allocation depends on the correct ranking of treatment effects, not merely their direction.

\begin{example}[Treatment Effect Sign Preservation]
\label{ex:budget_constraint}
% Suppose that $A \sim \text{Bernoulli}(0.5)$, $X \sim \text{Bernoulli}(0.5)$, and $\varepsilon_{S}, \varepsilon_{Y} \sim N(0, 1)$ independently. We specify the structural equation model as:
Let $A$ and $X$ be independent $\text{Bernoulli}(0.5)$ variables, and let $\varepsilon_{S}, \varepsilon_{Y} \sim N(0, 1)$ independently. We specify the structural equation model as:
\begin{align}
S &= (1-A)(-\alpha - \beta X) + \varepsilon_{S}, \notag\\
Y &= (1-A)(-\alpha + \beta X) + \varepsilon_{Y}, \notag
\end{align}
with the corresponding CATEs:
$$\tau_{S}(X) = \alpha + \beta X, \quad \tau_{Y}(X) = \alpha - \beta X,$$
where $\alpha > 0$ and $\beta > 0$.

% When $\alpha$ is sufficiently large, both treatment effects remain uniformly positive, thereby satisfying sign consistency. In the absence of budget constraints, both measures recommend universal treatment, resulting in perfect alignment. However, under budget constraints that permit treatment of only a fraction of the population, the policies diverge dramatically. The surrogate-based policy prioritizes individuals with larger values of $X$, whereas the outcome-based policy prioritizes those with smaller values of $X$. Although both treatment effects maintain identical signs, they exhibit completely opposite rankings across individuals, creating systematic misalignment in treatment allocation decisions.

When $\alpha$ is sufficiently large, both CATEs remain uniformly positive. As a result, without budget constraints, both the surrogate-optimal and outcome-optimal ITRs recommend treating everyone. 

However, under budget constraints, even under treatment-effect sign preservation, the expected outcome of the surrogate-optimal ITR can be substantially worse than that of the outcome-optimal ITR, and may even be worse than random assignment.
For instance, when the budget constraint permits treatment of $50\%$ of the population, the outcome-optimal ITR treats units with $X=0$, whereas the surrogate-optimal ITR selects units with $X=1$. Therefore, the surrogate-optimal ITR is even inferior to a random assignment rule.

This example highlights that treatment-effect sign preservation is insufficient for decision-making under budget constraints. What matters in this setting is not only the sign of treatment benefit, but also the ordering of individuals by the magnitude of their benefit.

% When $\alpha$ is sufficiently large, both CATEs remain uniformly positive, thereby leading both the corresponding surrogate-optimal and outcome-optimal ITRs to recommend treating everyone without budget constraints. 

% However, under budget constraints, even under the sign preservation condition, the expected outcome of the surrogate-optimal ITR can be severely suboptimal compared to the outcome-optimal ITR, sometimes performing even worse than random assignment. For instance, when the budget constraint permits treatment of $50\%$ of the population, the outcome-optimal ITR would treat units with $X < 0.5$, whereas the surrogate-optimal ITR would select units with $X > 0.5$. 
% This misalignment leads to an expected outcome inferior to a random assignment. 
% However, under budget constraints that permit treatment of only a fraction of the population, the ITRs diverge dramatically. The surrogate-optimal ITR prioritizes individuals with larger values of $X$, whereas the outcome-optimal ITR prioritizes those with smaller values of $X$. 

% This demonstrates that even sign preservation (the Assumption 4 in \cite{yang2023targeting}) cannot ensure adequate ITR alignment under realistic resource constraints, as treatment effect rankings can be completely reversed between the surrogate endpoint and the primary outcome.
\end{example}

These three examples show that several intuitive and seemingly strong notions of surrogate quality are insufficient for learning good ITRs. High observed correlation, strong dependence between potential outcomes, and even agreement in treatment-effect signs do not guarantee that the treatment rule induced by the surrogate performs well on the primary outcome. The core difficulty is that ITR learning depends on how well the surrogate preserves the treatment-relevant structure across individuals. In unconstrained settings, this concerns the sign of the treatment effect. Under budget constraints, it also concerns the relative magnitudes and ranking of treatment benefits. These observations motivate the need for surrogate evaluation measures that directly assess the quality of surrogate-induced treatment rules, which we develop in the next section.

\section{Surrogate Evaluation Measures for ITRs}
\label{sec_metirc}

Motivated by the failures illustrated in Section \ref{sec:3}, we now define evaluation measures that directly assess the value of surrogates through the quality of the treatment rules they induce on the primary outcome. 
We consider the unconstrained setting in Section 4.1 and the budget-constrained setting in Section 4.2.

% In this section, we develop a framework for evaluating surrogate endpoints in ITRs. Section 4.1 defines surrogate regret for unconstrained settings, which quantifies the loss from using surrogate-optimal ITRs instead of outcome-optimal ITRs. Section 4.2 extends our framework to budget-constrained settings by defining the $\lambda$-surrogate regret, and introduces two complementary measures:  $\lambda$-surrogate gain (with no treatment as the baseline), and $\lambda$-surrogate efficiency (with randomized treatment as the baseline). Together, these measures thoroughly evaluate the utility of the surrogate.
% a comprehensive assessment of surrogate quality.
% , enabling practitioners to make informed decisions about surrogate utility in individualized decision-making.

Under Assumption 1 and Assumption 2, the CATEs for both the primary outcome and the surrogate endpoint are identifiable from the observed data. Specifically, the conditional mean counterfactual outcome $E\left[Y(a) \mid X\right]$ can be identified by the outcome regression function $\mu_a(x) = E\left[Y \mid X=x, A=a\right]$ \citep{murphy2003optimal, robins2004optimal}. Similarly, for the surrogate, we identify $E\left[S{(a)} \mid X\right]$ by $\mu_{S,a}(x) = E\left[S \mid X=x, A=a\right]$. Consequently, the CATEs are identified by:
\begin{align}
\tau_Y(x) = \mu_{1}(x) - \mu_{0}(x), ~~ \tau_S(x) = \mu_{S,1}(x) - \mu_{S,0}(x). \notag
\end{align}

We first note that when there are no budget constraints, the ITR $\pi_Y = \mathbf{1}(\tau_Y(X) > 0)$ maximizes the expectation of the primary outcome $Y$, and ITR $\pi_S = \mathbf{1}(\tau_S(X) > 0)$ maximizes the expectation of the surrogate endpoint $S$. 

However, in many applications, only a limited fraction of the population can be treated because of budget, supply, or operational constraints \citep{bhattacharya2012inferring, matrajt2021vaccine, xu2024optimal}. We therefore consider ITRs that satisfy a budget constraint.
When at most a fraction $\lambda$ of the population can be treated, we consider $\lambda$-feasible ITR $\pi$ that satisfies $E\left[\pi(X)\right] \leq \lambda$. An ITR $\pi$ is $\lambda$-optimal if it maximizes the expected outcome among all $\lambda$-feasible ITRs. Let $y_{1-\lambda}$ and $s_{1-\lambda}$ denote the $(1-\lambda)$-quantiles of the treatment effect distributions $\tau_Y(X)$ and $\tau_S(X)$, respectively. Then it can be easily shown\footnote{{ For simplicity, we suppress tie-breaking at the threshold; when $P\{\tau_Y(X)=y_{1-\lambda}\}>0$ or $P\{\tau_S(X)=s_{1-\lambda}\}>0$, any appropriate randomized tie-breaking rule yields a $\lambda$-optimal ITR.}} that the following two ITRs are $\lambda$-optimal ITRs for the primary outcome and surrogate endpoint respectively:
\begin{equation}\label{eq:constr-opt-itr} 
    \begin{aligned}
        \pi_{Y,\lambda}(X) &= \mathbf{1}\{\tau_Y(X) > 0\} \cdot \mathbf{1}\{\tau_Y(X) > y_{1-\lambda}\},  \\
    \pi_{S,\lambda}(X) &= \mathbf{1}\{\tau_S(X) > 0\} \cdot \mathbf{1}\{\tau_S(X) > s_{1-\lambda}\}. 
    \end{aligned}
\end{equation}

In this section, we propose several measures to quantify the suboptimality and value of the surrogate-induced optimal ITR $\pi_S$ and $\pi_{S, \lambda}$ on the primary outcome.

\subsection{Surrogate Regret in the Unconstrained Setting}

A natural first question is how much primary-outcome value is lost when the surrogate-induced ITR \(\pi_S\) is used in place of the outcome-optimal ITR \(\pi_Y\). The following lemma formalizes a simple sufficient condition under which the two rules agree in value.

\begin{lemma} \label{ITR_CATE}
    $\pi_S$ is an optimal ITR for outcome $Y$, i.e., $E\left[ Y(\pi_S(X))\right] =E\left[ Y(\pi_Y(X))\right]$, 
if and only if $E \left[ \mathbf{1}(\tau_Y(X) < 0)\cdot \mathbf{1}(\tau_S(X) > 0)  \right] = E \left[ \mathbf{1}(\tau_Y(X) > 0)\cdot \mathbf{1}(\tau_S(X) \leq 0)  \right] =0$. 
\end{lemma}

Lemma \ref{ITR_CATE} shows that when the CATEs $\tau_S(X)$ and $\tau_Y(X)$ have the same sign, the surrogate-induced ITR is as good as the outcome-optimal ITR in terms of the expected primary outcome. However, when the signs differ, the surrogate-optimal ITR $\pi_S$ may result in suboptimality loss. 
To quantify this potential loss, we introduce the surrogate regret measure.
% However, when such sign preservation fails, the surrogate-optimal ITR $\pi_S$ may produce a lower expected outcome than the outcome-optimal ITR $\pi_Y$. To quantify this potential loss, we introduce the surrogate regret measure.

\begin{definition}[Surrogate Regret]
\label{def:psr}
The surrogate regret is defined as
\begin{equation}
    R = E\left[Y(\pi_Y(X))\right] - E\left[Y(\pi_S(X))\right],\notag
\end{equation}
which quantifies the expected outcome loss from using the surrogate-optimal ITR $\pi_S(X)$ instead of the outcome-optimal ITR $\pi_Y(X)$. 
\end{definition}

By the law of total expectation and noting that $\tau_Y(X) = E\left[Y(1) - Y(0) \mid X\right]$, the surrogate regret can be equivalently expressed as
\begin{equation}
\label{eq:psr_decomp}
    R = E\left\{\left[Y(1) - Y(0)\right] \cdot \left[\pi_Y(X) - \pi_S(X)\right]\right\} = E\left\{\tau_Y(X) \cdot \left[\pi_Y(X) - \pi_S(X)\right]\right\}.\notag
\end{equation}
Hence, surrogate regret is a weighted measure of sign disagreement between \(\tau_Y(X)\) and \(\tau_S(X)\), where the weight is the magnitude of the treatment effect on the primary outcome.

% This decomposition shows that $R$ depends on (i) the magnitude of the true treatment effect $\tau_Y(X)$ and (ii) the disagreement between the outcome-optimal and surrogate-optimal ITRs. 

% \begin{remark}
% Since policies disagree only when CATEs have opposite signs, we can write:
% \begin{equation}
%     R = E_P[|\tau_Y(X)| \cdot \mathbf{1}\left\{\text{sign}(\tau_Y(X)) \neq \text{sign}(\tau_S(X))\right\}].
% \end{equation}
% This highlights that Surrogate Regret captures both the frequency of disagreement and the magnitude of treatment effects where disagreement occurs.
% \end{remark}

In the unconstrained setting, this quantity captures the main discrepancy of interest, namely whether the surrogate and the primary outcome induce the same treatment decisions. In the budget-constrained setting, however, one must also account for the budget constraint. We next extend the surrogate regret to budget-constrained settings

% While surrogate regret quantifies the absolute loss from using a surrogate, it does not directly indicate whether the surrogate-optimal ITR outperforms baseline ITRs such as treating no one, or random assignment. Additionally, many applications have budget constraints that limit the population that can be treated. We therefore extend our framework to budget-constrained settings and define additional measures.

\subsection{Surrogate Regret with Budget Constraint}
\label{subsec:4.2}

In the budget-constrained setting, the first quantity of interest remains the loss from using the surrogate-induced rule rather than the outcome-optimal rule.

% Given a budget constraint $\lambda \in (0,1]$, an ITR $\pi$ is $\lambda$-feasible if $E\left[\pi(X)\right] \leq \lambda$. An ITR $\pi$ is $\lambda$-optimal if it maximizes the expected outcome among all $\lambda$-feasible ITRs. Let $y_{1-\lambda}$ and $s_{1-\lambda}$ denote the $(1-\lambda)$-quantiles of the treatment effect distributions $\tau_Y(X)$ and $\tau_S(X)$, respectively.
% % We use $y_{1-\lambda}$ to denote the $(1-\lambda)$-quantile of the treatment effect distribution $\tau_Y(X)$ and use $s_{1-\lambda}$ to denote the $(1-\lambda)$-quantile of the treatment effect distribution $\tau_S(X)$.

% As in the unconstrained case, we have the $\lambda$-outcome-optimal ITR $\pi_{Y,\lambda}$ and the $\lambda$-surrogate-optimal ITR $\pi_{S,\lambda}$:
% \begin{align}
%     \pi_{Y,\lambda}(X) = \mathbf{1}\{\tau_Y(X) > 0\} \cdot \mathbf{1}\{\tau_Y(X) > y_{1-\lambda}\}, \notag \\
%     \pi_{S,\lambda}(X) = \mathbf{1}\{\tau_S(X) > 0\} \cdot \mathbf{1}\{\tau_S(X) > s_{1-\lambda}\}. \notag
% \end{align}
% These ITRs treat individuals with the highest treatment effects, subject to the budget constraint. We now define surrogate regret for budget-constrained settings.
% % Having defined these constrained optimal ITRs, we can now extend the measure to the budget-constrained setting. 

\begin{definition}[$\lambda$-Surrogate Regret]
\label{def:lambda-sr}
Let $\lambda \in (0, 1]$ be the budget constraint, and let $\pi_{Y,\lambda}$ and $\pi_{S,\lambda}$ be the constrained optimal ITRs defined in eq.\eqref{eq:constr-opt-itr}. The $\lambda$-surrogate regret is defined as
\begin{equation}
    R(\lambda) = E\left[Y(\pi_{Y, \lambda}(X))\right] - E\left[Y(\pi_{S, \lambda}(X))\right].\notag
\end{equation}
When $\lambda = 1$, this reduces to the surrogate regret $R$ in Definition \ref{def:psr}.
\end{definition}

Similar to Definition \ref{def:psr}, the $\lambda$-surrogate regret can be equivalently expressed as
\begin{equation}
    R(\lambda) = E\left\{\tau_Y(X)\left[\pi_{Y, \lambda}(X) - \pi_{S, \lambda}(X)\right] \right\}.\notag
\end{equation}

The $\lambda$-surrogate regret quantifies the loss from using a $\lambda$-surrogate-optimal ITR instead of a $\lambda$-outcome-optimal ITR. 
However, regret alone does not reveal whether the surrogate-induced ITR achieves a practically meaningful benefit. A small regret can still correspond to a ITR with negligible value if the overall treatment benefit is itself small. This motivates the next two measures.

% However, 
% regret captures only the relative loss. Without quantifying the absolute benefit achieved, it is difficult to judge whether this gap is practically significant.
% % regret alone cannot assess surrogate quality. Without knowing the benefit achieved, we cannot interpret whether the performance gap is substantial or negligible. 
% We therefore develop complementary measures that capture the actual benefit from using the surrogate.

\subsection{$\lambda$-Surrogate Gain and $\lambda$-Surrogate Efficiency}

% We propose $\lambda$-surrogate gain to compare a $\lambda$-surrogate-optimal ITR against the no-treatment baseline, revealing the benefit gained from leveraging surrogate information.

To assess the absolute benefit achieved by the surrogate-induced ITR, we compare it with the ITR of treating no one.

\begin{definition}[$\lambda$-Surrogate Gain]
Let $\lambda \in (0, 1]$ be the budget constraint, the $\lambda$-surrogate gain is defined as:
$$G(\lambda) = E\left[Y(\pi_{S, \lambda}(X))\right] - E\left[Y(0)\right],$$
\end{definition}

Similar to Definition \ref{def:psr}, the $\lambda$-surrogate gain can be equivalently expressed as 
$$G(\lambda) = E\left[ \tau_Y(X)\, \pi_{S, \lambda}(X) \right].$$

The quantity \(G(\lambda)\) measures the improvement achieved by using the surrogate-induced \(\lambda\)-feasible ITR relative to no treatment. It complements \(R(\lambda)\) by quantifying the actual gain delivered by the surrogate-induced rule on the primary outcome.

To further capture the value of the surrogate, we also introduce the 
$\lambda$-surrogate efficiency, which quantifies the improvement over budget-matched random treatment rule. It reflects the informational value of the surrogate for targeting treatment under the budget constraint.

% {However, neither measure quantifies how much information the surrogate provides for decision-making.} To capture the informational value, we introduce $\lambda$-surrogate efficiency, which uses random treatment assignment as the baseline. This measure calculate the benefit of using the surrogate over random treatment.

% \subsection{$\lambda-$Surrogate Efficiency}

\begin{definition}[$\lambda$-Surrogate Efficiency]
Let $\lambda \in (0, 1]$ be the budget constraint, $\pi_{\lambda}(X)$ treats each individual independently with probability $\lambda$, regardless of $X$. The $\lambda$-surrogate efficiency is defined as:
$$V(\lambda) = E\left[Y(\pi_{S,\lambda}(X))\right] - E\left[Y(\pi_{\lambda}(X))\right].$$
\end{definition}

Similar to Definition \ref{def:psr}, this measure can be rewritten as:
$$V(\lambda) = E\left\{ \tau_Y(X) \left[ \pi_{S,\lambda}(X) - \lambda \right]\right\}.$$

The three budget-constrained measures capture different but related aspects of surrogate performance. Together, they provide complementary summaries of the quality of the treatment rules induced by the surrogate on the primary outcome.

% This formulation clearly shows that efficiency depends on how well the surrogate-optimal ITR allocates treatment to individuals with high treatment effects, as compared to random assignment with the budget $\lambda$.

% This formulation shows that efficiency depends on how well the $\lambda$-surrogate-optimal ITR prioritizes individuals with large positive treatment effects compared to a randomized treatment rule subject to the same budget $\lambda$.

% This rewritten form reveals that the efficiency gain stems from effectively reallocating the budget $\lambda$: assigning higher-than-random probabilities to individuals with large positive CATEs, and lower-than-random probabilities to those with small or negative CATEs.

\section{Estimation and Inference}
\label{sec:5}
In this section, we propose estimators for the surrogate evaluation measures in Sections \ref{sec_metirc}. All estimators follow a unified augmented inverse propensity weighting (AIPW) framework in Section 5.1, with specific adaptations for each measure. 
However, the non-smooth indicator functions in ITR definitions complicate the inference.
To establish asymptotic properties, we introduce some margin conditions.
Section 5.2 develops estimators for surrogate regret and $\lambda$-surrogate regret. Sections 5.3 and 5.4 construct estimators for $\lambda$-surrogate gain and $\lambda$-surrogate efficiency. 
% These estimators use adapted margin conditions and similar estimating functions but require weaker convergence conditions.

We first introduce notation for the subsequent analysis. For a random variable $Z$, we employ the essential supremum norm $\|Z\|_{\infty} = \inf\{M: P(|Z| > M) = 0\}$ and the $L_2$ norm $\|Z\|_2 = \{E(Z^2)\}^{1/2}$. These norms quantify errors of nuisance estimators, which we estimate via sample splitting. For the asymptotic analysis, we adopt the standard stochastic order notation: for  a random variable sequence $Z_N$, we write $Z_N = O_{P}(a_N)$ if, for any $\varepsilon > 0$, there exists finite $M > 0$ such that $P\left(|Z_N/a_N| > M\right) < \varepsilon$. Similarly, $Z_N = o_{P}(a_N)$ if $P\left(|Z_N/a_N| > \varepsilon\right) \to 0$ as $N \to \infty$.
%Specifically, we partition the primary dataset $\mathcal{D}_1$ into two subsets: $\mathcal{D}_{1,1}$, containing a randomly selected $30\%$ of observations, and $\mathcal{D}_{1,2}$, containing the remaining $70\%$. The subset $\mathcal{D}_{1,2}$, together with the auxiliary dataset $\mathcal{D}_2$, is used to train the plug-in nuisance parameters $\hat{\eta} = \left\{ \hat{\mu}_0(X), \hat{\mu}_1(X), \hat{e}(X), \hat{\pi}_Y(X), \hat{\pi}_S(X) \right\}$ and $\hat{\eta} = \left\{ \hat{\mu}_0(X), \hat{\mu}_1(X), \hat{e}(X), \hat{\pi}_{Y, \lambda}(X), \hat{\pi}_{S, \lambda}(X) \right\}$. The subset $\mathcal{D}_{1,1}$ serves as the test set, with $\mathcal{I}$ denoting the corresponding observation indices. The empirical expectation is defined as $E_{I_{1,1}}(f) = \frac{1}{|\mathcal{I}|} \sum_{i=1}^n f(A_i, X_i, Y_i) \mathbf{1}(i \in \mathcal{I})$, where $|\mathcal{I}|$ represents the cardinality of the test set.

\subsection{Constructing AIPW estimators}
\label{sec:5.1}

Consider the measures defined in Section \ref{sec_metirc}, including the surrogate regret $R$, the $\lambda$-surrogate regret $R(\lambda)$, the $\lambda$-surrogate gain $G(\lambda)$, and the $\lambda$-surrogate efficiency $V(\lambda)$. Formally, each of these measures can be expressed as a common target estimand $\theta$:
\begin{equation}
\theta = E\left[\phi(Y, A, X; \eta)\right], \notag
\end{equation} 
where $\phi(Y, A, X; \eta)$ is a specific estimating function depending on the observed data and a collection of nuisance parameters $\eta$. The specific form of $\phi$ varies according to the target measure $\theta$ and will be detailed in the subsequent subsections.

While the functional form of $\phi$ and the specific subset of nuisance parameters $\eta$ vary across different measures, all required nuisance parameters are from a common set:
\begin{equation}
\left\{\mu_0, \mu_1, \mu_{S,0}, \mu_{S,1}, e, \pi_Y, \pi_S, \pi_{Y,\lambda}, \pi_{S,\lambda}, y_{1-\lambda}, s_{1-\lambda}\right\}. \notag
\end{equation}

Here, $\mu_a$ and $\mu_{S,a}$ denote the conditional mean functions for the primary outcome and the surrogate endpoint, respectively; $e$ denotes the propensity score. The set includes parameters for both unconstrained and budget-constrained scenarios: $\pi_Y$ and $\pi_S$ represent the outcome-optimal ITR and surrogate-optimal ITR, whereas $\pi_{Y,\lambda}$ and $\pi_{S,\lambda}$ denote the $\lambda$-optimal ITRs. These $\lambda$-optimal ITRs are jointly determined by the CATEs $\tau_Y$ and $\tau_S$ and their respective quantiles $y_{1-\lambda}$ and $s_{1-\lambda}$.

% We formulate our AIPW estimators using a unified representation. The target estimand can be expressed as
% \begin{equation}
% \theta = E\left[\phi(Y, A, X; \eta)\right], \notag
% \end{equation} 
% where $\phi(Y, A, X; \eta)$ denote a general estimating function that depends on the observed data $(Y, A, X)$ and a collection of nuisance parameters $\eta$, $\theta$ represents a specific measure such as the surrogate regret $R$, the $\lambda$-surrogate regret $R(\lambda)$, the $\lambda$-surrogate gain $G(\lambda)$, or the $\lambda$-surrogate efficiency $V(\lambda)$, depending on the specific form of $\phi$.

% We define the full collection of nuisance parameters required for the budget-constrained setting as
% \begin{equation}
% \eta \subset \{\mu_0, \mu_1, \mu_{S,0}, \mu_{S,1}, e, \pi_{Y,\lambda}, \pi_{S,\lambda}, y_{1-\lambda}, s_{1-\lambda}\}. \notag
% \end{equation}
% Here, $\mu_a(X)$ and $\mu_{S,a}(X)$ denote the conditional mean counterfactual outcomes for {\color{blue}the primary outcome and the surrogate endpoint}, and $e(X)$ is the propensity score. The parameters $y_{1-\lambda}$ and $s_{1-\lambda}$ represent the quantiles of the treatment effects $\tau_Y(X)$ and $\tau_S(X)$, respectively, which determine the optimal ITRs $\pi_{Y,\lambda}$ and $\pi_{S,\lambda}$ as defined in Section \ref{subsec:4.2}.

\begin{definition}[Sample-Splitting Estimator]
\label{def:estimator}
Recall that we have the outcome dataset $\mathcal{D}_1 = \{(A_i, X_i, Y_i): i \in I_1\}$. First, we randomly partition the index set $I_1$ into two equal-sized disjoint subsets, $I_{1,1}$ and $I_{1,2}$, and let $n = \left|I_{1,1}\right|$ denote the size of the estimation sample. Next, we utilize the combination of the nuisance estimation fold $\mathcal{D}_{1,2} = \{(A_i, X_i, Y_i): i \in I_{1,2}\}$ and the surrogate dataset $\mathcal{D}_2$ to estimate the nuisance parameters $\hat{\eta}$, which vary across different target measures as will be explained later. Finally, the estimator $\hat{\theta}$ is computed on the main sample $\mathcal{D}_{1,1} = \{(A_i, X_i, Y_i): i \in I_{1,1}\}$ as:
% Recall that we have the outcome dataset $\mathcal{D}_1 = \{(A_i, X_i, Y_i): i \in I_1\}$. First, we randomly partition the index set $I_1$ into two equal-sized disjoint subsets, $I_{1,1}$ and $I_{1,2}$, and let $n = |I_{1,1}|$ denote the size of the estimation sample., yielding $\mathcal{D}_{1,1} = \{(A_i, X_i, Y_i): i \in I_{1,1}\}$ and $\mathcal{D}_{1,2} = \{(A_i, X_i, Y_i): i \in I_{1,2}\}$. Next, we utilize the combination of $\mathcal{D}_{1,2}$ and the surrogate dataset $\mathcal{D}_2$ to estimate the nuisance parameters $\hat{\eta}$, which vary by target measure as detailed later. Finally, the estimator $\hat{\theta}$ is computed on the main sample $\mathcal{D}_{1,1}$ as:
\begin{equation}
\hat{\theta} = E_{I_{1,1}} \left[\phi(Y_i, A_i, X_i; \hat{\eta})\right] = n^{-1} \sum_{i \in I_{1, 1}} \phi(Y_i, A_i, X_i; \hat{\eta}).
\notag
\end{equation}
\end{definition}
For simplicity, our main text focuses on  the sample splitting procedure in Definition \ref{def:estimator}. Its finite-sample performance can be improved by implementing a cross-fitting procedure \citep{chernozhukov2018double}. Cross-fitting randomly splits the data into $K$ disjoint folds, using the data in all but one fold to estimate nuisance parameters and applying them only to the specific hold-out fold, and finally averaging the estimates across all folds. The extension is detailed in Supplementary Material Section \ref{appendix: B1}. In scenarios where only a single dataset $\mathcal{D}$ is available, we adapt the proposed estimation framework as described in Supplementary Material Section \ref{appendix: B2}.

\subsection{Surrogate Regret and $\lambda$-Surrogate Regret}

We propose an AIPW estimator for the surrogate regret $R$. The main challenge in estimating $R$ is its non-smoothness. Without additional assumptions, $R$ is not pathwise differentiable due to the indicator functions in the outcome-optimal ITR $\pi_Y$ and the surrogate-optimal ITR $\pi_S$. When the CATEs $\tau_Y(X)$ or $\tau_S(X)$ have point mass at decision thresholds, the resulting discontinuities complicate asymptotic analysis \citep{Levis2024}.

To address this, we introduce margin conditions that limit the probability mass near the thresholds. These conditions ensure desirable convergence properties for the proposed estimators.

\begin{assumption}[margin conditions]
\label{ass:margin-3}
% \hfill
For the CATEs $\tau_Y$ and $\tau_S$, we assume:
\vspace*{-8pt}
\begin{enumerate}
\item[(a)] For some $\alpha_1 > 0$, $P(|\tau_Y(X)| \leq t) = O(t^{\alpha_1})$ for all $t \ge 0$.
\item[(b)] For some $\alpha_2 > 0$, $P(|\tau_S(X)| \leq t) = O(t^{\alpha_2})$ for all $t \ge 0$.
\item[(c)] For some $\beta_1 > 0$, $P(|\tau_Y(X) - y_{1-\lambda}| \leq t) = O(t^{\beta_1})$ for all $t \ge 0$.
\item[(d)] For some $\beta_2 > 0$, $P(|\tau_S(X) - s_{1-\lambda}| \leq t) = O(t^{\beta_2})$ for all $t \ge 0$.
\end{enumerate}
\end{assumption}
% \renewcommand{\theassumption}{3\alph{assumption}}
% \begin{assumption}[margin condition for $\tau_Y$]
% \label{ass:margin-3a}
% For some $\alpha_1 > 0$, $P(|\tau_Y(X)| \leq t) = O(t^{\alpha_1})$ for all $t \ge 0$.
% \end{assumption}

% \begin{assumption}[margin condition for $\tau_S$]\renewcommand{\theassumption}{3b}
% \label{ass:margin-3b}
% For some $\alpha_2 > 0$, $P(|\tau_S(X)| \leq t) = O(t^{\alpha_2})$ for all $t \ge 0$.
% \end{assumption}

% \begin{assumption}[margin condition for $\tau_Y$ around $y_{1-\lambda}$]\renewcommand{\theassumption}{3c}
% \label{ass:margin-3c}
% For some $\beta_1 > 0$, $P(|\tau_Y(X) - y_{1-\lambda}| \leq t) = O(t^{\beta_1})$ for all $t \ge 0$.
% \end{assumption}

% \begin{assumption}[margin condition for $\tau_S$ around $s_{1-\lambda}$]\renewcommand{\theassumption}{3d}
% \label{ass:margin-3d}
% For some $\beta_2 > 0$, $P(|\tau_S(X) - s_{1-\lambda}| \leq t) = O(t^{\beta_2})$ for all $t \ge 0$.
% \end{assumption}

% \setcounter{assumption}{3}

% Redefine the display format of the assumption counter with letters
\renewcommand{\theassumption}{\arabic{assumption}}
The conditions in Assumption \ref{ass:margin-3} are analogous to those in the classification literature \citep{tsybakov2004optimal, audibert2007fast}, and in causal inference and ITR estimation \citep{qian2011performance, luedtke2016statistical, kennedy2020optimal, kallus2022harm, dadamo2021doubly, levis2023covariate, benmichael2024policy}. 
These conditions ensure that the distributions of $\tau_Y$ and $\tau_S$ do not concentrate too heavily. Conditions (a)-(b) rule out excessive density near zero, while conditions (c)-(d) prevent concentration near the quantiles $y_{1-\lambda}$ and $s_{1-\lambda}$.
% These conditions assert that the distributions of $\tau_Y$ and $\tau_S$ do not concentrate too heavily around zero or their respective quantiles $y_{1-\lambda}$ and $s_{1-\lambda}$. 
Practically, this assumption rules out scenarios where the CATE concentrates heavily at the decision boundary, which would make the hard thresholding in optimal ITRs particularly unstable.

To estimate the regret $R$, we construct an estimating function based on the AIPW structure:
\[
\phi_R(Y, A, X;\, \eta) = \left[\pi_Y(X) - \pi_S(X)\right]
\left\{\frac{A}{e(X)} - \frac{1-A}{1-e(X)} \right\}
\left[Y - \mu_{A}(X)\right] + \tau_Y(X)\left[\pi_Y(X) - \pi_S(X)\right],
\]
where $\eta = \left( \mu_0, \mu_1, \mu_{S,0}, \mu_{S,1}, e, \pi_Y, \pi_S \right)$ represents the corresponding collection of nuisance parameters. 
% This function satisfies $E\left[\phi_R(Y, A, X;\, \eta)\right] = R$ by the law of total expectation.
% { and possesses a crucial second-order bias property, as established in the following result}.
% Since the true nuisance parameters are unknown, 
We estimate these parameters using the sample-splitting strategy outlined in Definition \ref{def:estimator}. Specifically, the propensity score estimator $\hat{e}$ and surrogate regression functions $\hat{\mu}_{S,a}$ ($a=0,1$) are estimated using $\mathcal{D}_2$, while the outcome regression functions $\hat{\mu}_a$ ($a=0,1$) are obtained from $\mathcal{D}_{1,2}$. Based on these, we construct the CATE estimators as $\hat{\tau}_Y(X) = \hat{\mu}_1(X) - \hat{\mu}_0(X)$ and $\hat{\tau}_S(X) = \hat{\mu}_{S,1}(X) - \hat{\mu}_{S,0}(X)$. The corresponding estimated ITRs are then given by:
\begin{equation}
\hat{\pi}_{Y}(X) = \mathbf{1}\{\hat{\tau}_Y(X) > 0\},\ \hat{\pi}_{S}(X) = \mathbf{1}\{\hat{\tau}_S(X) > 0\}, \notag
\end{equation}
According to Definition \ref{def:estimator}, with $\hat{\eta}$ estimated, our proposed estimator is:
$$\hat{R} = E_{I_{1,1}}\left[ \phi_R(Y, A, X;\, \hat{\eta}) \right].$$
% where, consistent with Definition \ref{def:estimator}, $E_{I_{1,1}}(f) = n^{-1} \sum_{i \in I_{1, 1}} f(A_i, X_i, Y_i)$ denotes the empirical expectation over $\mathcal{D}_{1,1}$. 
The following lemma characterizes the bias introduced by substituting $\eta$ with $\hat{\eta}$.

\begin{lemma}
\label{lemma2}
    % Let $\hat{\eta}$ be a collection of estimators for the unknown nuisance parameters $\eta$. Then 
    We have
    \begin{equation}
        \begin{split}
            &E\left[\phi_R(Y, A, X;\, \hat{\eta}) - \phi_R(Y, A, X;\, \eta)\mid X\right] \\
            &=\, \underbrace{\left[\hat{\pi}_{Y}(X) - \hat{\pi}_{S}(X)\right] \left\{ 
            \frac{\hat{e}(X) - e(X)}{\hat{e}(X)}\left[\hat{\mu}_1(X) - \mu_1(X)\right] + 
            \frac{\hat{e}(X) - e(X)}{1 - \hat{e}(X)}\left[\hat{\mu}_0(X) - \mu_0(X)\right]\right\}}_{I} \\
            &\quad + \underbrace{\tau_Y(X)
            \left\{ [\hat{\pi}_{Y}(X) - \pi_{Y}(X)] - 
            [\hat{\pi}_{S}(X) - \pi_{S}(X)] \right\}}_{II}. \notag
        \end{split}
    \end{equation}
\end{lemma}

This lemma decomposes the bias in our estimating function. Term I is the product of estimation errors in the propensity score and outcome regression functions. This product structure makes Term I asymptotically negligible if either component is consistently estimated. Term II, however, is a first-order bias term that requires careful control through margin conditions. Building on this decomposition, we establish the convergence rate of $\hat{R}$ to $R$, and derive sufficient conditions for asymptotic normality.
\begin{theorem}
\label{thm: thm1}
Assume that $\|\tau_Y(X)\|_{\infty} < \infty$ and $E(Y^2) < \infty$.
% Assume that there exists $M > 0$ such that $P_X(|\tau_Y(X)| \leq M) = 1$.
Moreover, assume that $\|\hat{\mu}_0 - \mu_0\|_2 + \|\hat{\mu}_1 - \mu_1\|_2 + \|\hat{\mu}_{S, 0} - \mu_{S, 0}\|_2 + \|\hat{\mu}_{S, 1} - \mu_{S, 1}\|_2 +
\|\hat{e} - e\|_2 = o_{P}(1)$. Then, under the $\text{Assumption } \ref{ass:margin-3}$(a) and $\ref{ass:margin-3}$(b), 
\[
\hat{R} - R = O_{P}(n^{-1/2} + D_{1, n} + D_{2, n} + D_{3, n}),
\]

where 
\[D_{1, n} = \left\|\hat{e} - e\right\|_2 \cdot 
    \left(\left\|\hat{\mu}_1 - \mu_1\right\|_2+ \left\|\hat{\mu}_0 - \mu_0\right\|_2\right),\] 
\[
D_{2, n} = \| \tau_Y - \hat{\tau}_Y \|^{1 + \alpha_1}_{\infty},\text{ and } D_{3, n} = \|\tau_S - \hat{\tau}_S\|_{\infty}^{\alpha_2}
.\]
\end{theorem}

\begin{proposition}
If $D_{1, n} + D_{2, n} + D_{3, n} = o_{P}(n^{-1/2})$, then
\[
n^{1/2}\left(\hat{R} - R\right) \xrightarrow{d} N(0, \sigma_R^2),
\]
where 
$\sigma_R^2 = \operatorname{Var}\left[\phi_R(Y, A, X;\, \eta)\right]$.
\end{proposition}

% The convergence rate decomposition reveals several key insights. 
In this decomposition, the term $D_{1,n}$ shares a similar error-product structure as term I in Lemma \ref{lemma2}. Consequently, it vanishes asymptotically if either the propensity score or the outcome regression is consistently estimated. The term $D_{2,n}$ reflects the complexity inherent in estimating optimal values, as previously studied by 
\citet{luedtke2016statistical}. Finally, $D_{3,n}$ is a bias term unique to our surrogate evaluation framework. It arises from the interaction $\tau_Y \cdot (\hat{\pi}_S - \pi_S)$ between the primary outcome effect and the error in estimating the surrogate-optimal ITR, and vanishes when the margin parameter $\alpha_2$ is sufficiently large.

The analysis of the $\lambda$-surrogate regret follows a similar theoretical foundation, with additional quantile estimation. We define the corresponding estimating function as follows:
\begin{equation}
    \begin{split}
        \phi_{R, \lambda}(Y, A, X;\, \eta) = 
        \,&\tau_Y(X)
        \left[\pi_{Y, \lambda}(X) - \pi_{S, \lambda}(X)\right] + 
        \\
        &\left[\pi_{Y, \lambda}(X) - \pi_{S, \lambda}(X)\right]
        \left\{\frac{A}{e(X)} - \frac{1-A}{1-e(X)} \right\}
        \left[Y - \mu_A(X)\right],
    \end{split}\notag
\end{equation}
where $\lambda \in (0, 1]$, $\eta = \left( \mu_0, \mu_1, \mu_{S,0}, \mu_{S, 1}, e, \pi_{Y, \lambda}, \pi_{S, \lambda}, y_{1-\lambda}, s_{1-\lambda} \right)$ represents the nuisance parameters for the budget-constrained setting. 

The estimation of $\eta$  involves shared nuisance parameters as in the unconstrained case, including the propensity score estimator $\hat{e}$, regression functions $\hat{\mu}_a,\ \hat{\mu}_{S, a}$ ($a=0,1$), from which the CATE estimators $\hat{\tau}_Y$ and $\hat{\tau}_S$ are derived. These components are obtained via the sample-splitting strategy described previously. With these components in place, we proceed to estimate the thresholding parameters required for the budget constraint. Specifically, $\hat{y}_{1-\lambda}$ is estimated as the empirical $(1-\lambda)$-quantile of $\hat{\tau}_Y$:
\begin{equation}
\hat{y}_{1-\lambda} = \inf\left\{t \in \mathbb{R}: \left|I_2\right|^{-1}\sum_{i \in I_2} \mathbf{1}\{\hat{\tau}_Y(X_i) \leq t\} \geq 1-\lambda\right\}. \notag
\end{equation}
Conversely, as $\hat{\tau}_S$ is constructed using $\mathcal{D}_2$, we estimate its corresponding threshold $\hat{s}_{1-\lambda}$ utilizing $\mathcal{D}_{1,2}$:
\begin{equation}
\hat{s}_{1-\lambda} = \inf\left\{t \in \mathbb{R}: \left|I_{1,2}\right|^{-1}\sum_{i \in I_{1,2}} \mathbf{1}\{\hat{\tau}_S(X_i) \leq t\} \geq 1-\lambda\right\}. \notag
\end{equation}
The resulting estimated ITR is given by:
$$
\hat{\pi}_{Y,\lambda}(X) = \mathbf{1}\{\hat{\tau}_Y(X) > \hat{y}_{1-\lambda}\}\cdot\mathbf{1}\{\hat{\tau}_Y(X) > 0\},\ 
\hat{\pi}_{S,\lambda}(X) = \mathbf{1}\{\hat{\tau}_S(X) > \hat{s}_{1-\lambda}\}\cdot\mathbf{1}\{\hat{\tau}_S(X) > 0\}. \notag
$$
% with $\hat{s}_{1-\lambda}$ and $\hat{\pi}_{S,\lambda}(X)$ defined analogously. 
Finally, combining these threshold estimators with the previously obtained nuisance parameters, the proposed estimator of $R(\lambda)$ is formulated as:
\[
\hat{R}(\lambda) = E_{I_{1,1}}\left[ \phi_{R, \lambda}(Y, A, X;\, \hat{\eta}) \right].
\]
\begin{theorem}
\label{thm: thm2}
Assume that $\|\tau_Y(X)\|_{\infty} < \infty$ and $E(Y^2) < \infty$. Moreover, assume that $\|\hat{\mu}_0 - \mu_0\|_2 + \|\hat{\mu}_1 - \mu_1\|_2  + \|\hat{\mu}_{S, 0} - \mu_{S, 0}\|_2 + \|\hat{\mu}_{S, 1} - \mu_{S, 1}\|_2 + \|\hat{e} - e\|_2
+ |y_{1-\lambda} - \hat{y}_{1-\lambda}| + |s_{1-\lambda} - \hat{s}_{1-\lambda}|
= o_{P}(1)$. Then, for any $\lambda \in (0, 1]$, under the $\text{Assumption } \ref{ass:margin-3}$,
\[
\hat{R}(\lambda) - R(\lambda) = O_{P}(n^{-1/2} + D_{1, n} + D_{2, n} + D_{3, n} + D_{4, n} + D_{5, n}),
\]

where $D_{1, n}$, $D_{2, n}$, and $D_{3, n}$ are defined as in Theorem \ref{thm: thm1},
\[D_{4, n} = \left(\|\tau_Y - \hat{\tau}_Y\|_{\infty} + 
        |y_{1-\lambda} - \hat{y}_{1-\lambda}|\right)^{\beta_1}\text{,}
\]
\[
D_{5, n} = \left(\|\tau_S - \hat{\tau}_S\|_{\infty} + 
        |s_{1-\lambda} - \hat{s}_{1-\lambda}|\right)^{\beta_2}.
\]
\end{theorem}

\begin{proposition}
If, in addition, $D_{1, n} + D_{2, n} + D_{3, n} + D_{4, n} + D_{5, n} = o_{P}(n^{-1/2})$, then
\[
n^{1/2}\left[\hat{R}(\lambda) - R(\lambda)\right] \xrightarrow{d} N(0, \sigma_R^2(\lambda)),
\]
where $\sigma_R^2(\lambda) = \operatorname{Var}\left[
\phi_{R, \lambda}(Y, A, X;\, \eta)\right]$.
\end{proposition}

The budget-constrained case introduces two additional bias terms, $D_{4,n}$ and $D_{5,n}$, which reflect the estimation errors in both the CATEs ($\|\tau_Y - \hat{\tau}_Y\|_\infty$, $\|\tau_S - \hat{\tau}_S\|_\infty$) and the quantile estimation ($|y_{1-\lambda} - \hat{y}_{1-\lambda}|$, $|s_{1-\lambda} - \hat{s}_{1-\lambda}|$). These terms become negligible when the margin parameters $\beta_1$ and $\beta_2$ are sufficiently large.

\subsection{$\lambda$-Surrogate Gain}

We now develop the estimator for surrogate gain using the same framework. The corresponding estimating function is:
\begin{equation}
\begin{split}
\phi_{G, \lambda}(Y, A, X;\, \eta) &= \pi_{S, \lambda}(X)
\left\{\frac{A}{e(X)} - \frac{1-A}{1-e(X)} \right\}
\left[Y - \mu_A(X)\right] + \tau_Y(X)\,\pi_{S, \lambda}(X),
\end{split}\notag
\end{equation}
where $\lambda \in (0, 1]$, $\eta = \left( \mu_0, \mu_1, \mu_{S,0}, \mu_{S, 1}, e, \pi_{S, \lambda}, s_{1-\lambda} \right)$. Our estimator is:
\[\
\hat{G}(\lambda) = E_{I_{1,1}}\left[ \phi_{G, \lambda}(Y, A, X;\, \hat{\eta}) \right].
\]
Since $\hat{G}(\lambda)$ has a functional form similar to $\hat{R}(\lambda)$ but involves fewer nuisance parameters, its bias analysis proceeds analogously to that of Lemma \ref{lemma2} and Theorem \ref{thm: thm2}. Detailed proofs are provided in Supplementary Material Section \ref{appendix: C}.
% The bias analysis for surrogate gain follows a similar structure as in the surrogate-regret case. Specifically, the bias decomposition involves products of estimation errors for nuisance functions. The detailed bias characterization is provided in Supplementary Material Section C.

\begin{theorem}
\label{thm: thm3}
Assume that $\|\tau_Y(X)\|_{\infty} < \infty$ and $E(Y^2) < \infty$.
% Assume that there exists $M > 0$ such that $P_X(|\tau_Y(X)| \leq M) = 1$. 
Moreover, assume that $\|\hat{\mu}_0 - \mu_0\|_2 + \|\hat{\mu}_1 - \mu_1\|_2 + \|\hat{\mu}_{S, 0} - \mu_{S, 0}\|_2 + \|\hat{\mu}_{S, 1} - \mu_{S, 1}\|_2 + 
\|\hat{e} - e\|_2 
+ |s_{1-\lambda} - \hat{s}_{1-\lambda}| 
= o_{P}(1)$ Then, for $\lambda \in (0, 1]$, under $\text{Assumption } \ref{ass:margin-3}(b) \text{ and }  \ref{ass:margin-3}(d)$,
\[
\hat{G}(\lambda) - G(\lambda) = O_{P}(n^{-1/2} + D_{1, n} + D_{3, n} + D_{5, n})
\]
where $D_{1,n}, D_{3, n}\text{, and } D_{5, n}$ are defined as in Theorem \ref{thm: thm1} and Theorem \ref{thm: thm2}.
\end{theorem}

\begin{proposition}
If, in addition, $D_{1, n} + D_{3, n} + D_{5, n} = o_{P}(n^{-1/2})$, then
\[
n^{1/2}\left[\hat{G}(\lambda) - G(\lambda)\right] \xrightarrow{d} N(0, \sigma_G^2(\lambda)),
\]
where 
$\sigma_G^2(\lambda) = \operatorname{Var}\left[\phi_{G, \lambda}(Y, A, X;\, \eta)\right]$.
\end{proposition}

\subsection{$\lambda$-Surrogate Efficiency}

% We complete our theoretical analysis by examining
Finally, we consider the $\lambda$-surrogate efficiency, which compares the performance of the $\lambda$-surrogate-optimal ITR $\pi_{S, \lambda}$ over the randomized treatment rule $\pi_\lambda$. The corresponding estimating function is:
\[
\phi_{V, \lambda}(Y, A, X;\, \eta) = \left(\pi_{S, \lambda}(X) - \lambda\right)
\left\{\frac{A}{e(X)} - \frac{1-A}{1-e(X)} \right\}
\left[Y - \mu_A(X)\right] + \tau_Y(X)
\left(\pi_{S, \lambda}(X) - \lambda\right),
\]
where $\lambda \in (0, 1]$, $\eta = \left( \mu_0, \mu_1, \mu_{S,0}, \mu_{S, 1}, e, \pi_{S, \lambda}, s_{1-\lambda} \right)$. Our estimator is \[\hat{V}(\lambda) = E_{I_{1,1}}\left[ \phi_{V, \lambda}(Y, A, X;\, \hat{\eta}) \right].
\]

The bias analysis for the $\lambda$-surrogate efficiency follows the previous approach, with details provided in the Supplementary Material Section \ref{appendix: C}.
Its convergence properties mirror those of the surrogate gain estimator, as shown in the following theorem.

\begin{theorem}
\label{thm: thm4}
Assume that $\|\tau_Y(X)\|_{\infty} < \infty$ and $E(Y^2) < \infty$.
% Assume that there exists $M > 0$ such that $P_X(|\tau_Y(X)| \leq M) = 1$. 
Moreover, assume that $\|\hat{\mu}_0 - \mu_0\|_2 + \|\hat{\mu}_1 - \mu_1\|_2 + \|\hat{\mu}_{S, 0} - \mu_{S, 0}\|_2 + \|\hat{\mu}_{S, 1} - \mu_{S, 1}\|_2 + 
\|\hat{e} - e\|_2 + |s_{1-\lambda} - \hat{s}_{1-\lambda}|  = o_{P}(1)$. Then, for $\lambda \in (0, 1]$, under $\text{Assumption } \ref{ass:margin-3}(b) \text{ and }  \ref{ass:margin-3}(d)$, 
\[
\hat{V}(\lambda) - V(\lambda) = O_{P}(n^{-1/2} + D_{1, n} + D_{3, n} + D_{5, n}),
\]
where $D_{1,n}, D_{3, n}\text{, and }D_{5, n}$ are defined as in Theorem \ref{thm: thm1} and Theorem \ref{thm: thm2}.
\end{theorem}

Like the surrogate gain estimator, surrogate efficiency estimation requires only three bias terms, 
% reflecting the fact that it involves only the
since the surrogate efficiency estimator depends only on the surrogate-optimal ITR, not the outcome-optimal ITR.
Under regularity conditions, we obtain asymptotic normality for the surrogate efficiency estimator.

\begin{proposition}
If $D_{1, n} + D_{3, n} + D_{5, n} = o_{P}(n^{-1/2})$, then
\[
n^{1/2}\left[\hat{V}(\lambda) - V(\lambda)\right] \xrightarrow{d} N(0, \sigma_V^2(\lambda)),
\]
where $\sigma_V^2(\lambda) = \operatorname{Var}\left[
\phi_{V, \lambda}(Y, A, X;\, \eta)\right]$.
\end{proposition}

\section{Experiments}
\label{sec:experiments}

In this section, we assess the finite-sample performance of the proposed estimators $\hat{R}$, $\hat{R}(\lambda)$, $\hat{V}(\lambda)$, and $\hat{G}(\lambda)$ across different quantile thresholds $\lambda$ through both simulation studies and a real data analysis. The  Python codes are provided in the supplementary file.

\subsection{Simulation Experiments}

\begin{table}
    \centering
    \small
    \setlength{\tabcolsep}{2.7pt}  % Reduce column spacing
    \begin{tabular}{lc|ccc|ccc|ccc|ccc}
    \hline
    & & \multicolumn{3}{c|}{$m = 500$} & \multicolumn{3}{c|}{$m = 1000$} & \multicolumn{3}{c|}{$m = 2000$} & \multicolumn{3}{c}{$m = 3000$} \\
    & $\lambda$ & Bias & SD & CP95 & Bias & SD & CP95 & Bias & SD & CP95 & Bias & SD & CP95 \\
    \hline
    $\hat{R}$
    & / & -0.0005 & 0.0551 & 0.9392 & 0.0003 & 0.0373 & 0.9472 & 0.0014 & 0.0264 & 0.9463 & 0.0019 & 0.0214 & 0.9486 \\
    \hline
    \multirow{4}{*}{$\hat{R}(\lambda)$} 
    & 0.1 & -0.0023 & 0.0387 & 0.9417 & -0.0023 & 0.0267 & 0.9460 & -0.0015 & 0.0192 & 0.9448 & -0.0015 & 0.0156 & 0.9474 \\
    & 0.2 & -0.0024 & 0.0494 & 0.9384 & -0.0014 & 0.0348 & 0.9485 & 0.0001 & 0.0250 & 0.9479 & 0.0004 & 0.0204 & 0.9510 \\
    & 0.3 & -0.0026 & 0.0540 & 0.9393 & -0.0015 & 0.0372 & 0.9454 & -0.0002 & 0.0264 & 0.9487 & 0.0002 & 0.0214 & 0.9494 \\
    & 0.4 & -0.0022 & 0.0549 & 0.9389 & -0.0014 & 0.0373 & 0.9466 & -0.0002 & 0.0264 & 0.9486 & 0.0002 & 0.0214 & 0.9495 \\
    \hline
    \multirow{4}{*}{$\hat{V}(\lambda)$} 
    & 0.1 & 0.0014 & 0.0293 & 0.9387 & 0.0013 & 0.0202 & 0.9442 & 0.0011 & 0.0143 & 0.9464 & 0.0014 & 0.0117 & 0.9472 \\
    & 0.2 & 0.0012 & 0.0372 & 0.9383 & 0.0008 & 0.0262 & 0.9437 & 0.0007 & 0.0186 & 0.9504 & 0.0014 & 0.0153 & 0.9482 \\
    & 0.3 & 0.0010 & 0.0416 & 0.9389 & 0.0001 & 0.0285 & 0.9464 & 0.0001 & 0.0199 & 0.9484 & 0.0007 & 0.0163 & 0.9477 \\
    & 0.4 & 0.0004 & 0.0456 & 0.9385 & -0.0006 & 0.0307 & 0.9463 & -0.0010 & 0.0214 & 0.9467 & -0.0003 & 0.0176 & 0.9464 \\
    \hline  
    \multirow{4}{*}{$\hat{G}(\lambda)$} 
    & 0.1 & 0.0006 & 0.0308 & 0.9367 & 0.0008 & 0.0213 & 0.9444 & 0.0008 & 0.0151 & 0.9469 & 0.0012 & 0.0124 & 0.9485 \\
    & 0.2 & -0.0003 & 0.0404 & 0.9374 & -0.0002 & 0.0290 & 0.9435 & 0.0003 & 0.0207 & 0.9481 & 0.0008 & 0.0170 & 0.9490 \\
    & 0.3 & -0.0014 & 0.0447 & 0.9411 & -0.0014 & 0.0314 & 0.9450 & -0.0007 & 0.0221 & 0.9495 & -0.0001 & 0.0180 & 0.9486 \\
    & 0.4 & -0.0028 & 0.0456 & 0.9419 & -0.0027 & 0.0316 & 0.9454 & -0.0019 & 0.0221 & 0.9503 & -0.0013 & 0.0180 & 0.9488 \\
    \hline
    \end{tabular}
    \caption{Simulation results for different estimators across budget values and sample sizes.}
    \label{tab:simulation_results}
    \end{table}

Throughout the simulation, covariates $X = (X_1, X_2)^T$ are generated from a bivariate normal distribution $N(0, I_2)$, where $I_2$ is the $2 \times 2$ identity matrix. The treatment indicator $A$ follows a Bernoulli distribution with success probability $P(A = 1\mid X) = \text{expit}(0.1X_1 + 0.1X_2)$, where $\text{expit}(x) = \exp(x)/\{1 + \exp(x)\}$. Sample sizes $m \in \{500, 1000, 2000, 3000\}$ are considered.

We consider binary outcomes, where both the primary outcome and the surrogate endpoint take values in $\{0, 1\}$. For \( a \in \{0,1\} \), the potential outcomes $Y(a)$ and $S(a)$ are generated independently from Bernoulli distributions with conditional probabilities:
\begin{align*}
P(Y(1)=1\mid X)&= \text{expit}(0.1X_1 + 0.1X_2 + 0.2), \\
P(Y(0)=1\mid X)&= \text{expit}(0.2X_1^2 + 0.2X_2 + 0.1), \\
P(S(1)=1\mid X)&= \text{expit}(0.3X_1 + 0.2X_2 + 0.2), \\
P(S(0)=1\mid X)&= \text{expit}(0.2X_1^2 + 0.1X_2 + 0.1). 
\end{align*}
% \begin{align*}
%     & $Y(1) \sim \text{Bernoulli}\{\mu_1^Y(X)\}$ with $\mu_1^Y(X) = \text{expit}(0.3X_1 + 0.1X_2)$\\
%      & $Y(0) \sim \text{Bernoulli}\{\mu_0^Y(X)\}$ with $\mu_0^Y(X) = \text{expit}(0.5X_1 + 0.3X_2)$\\
%   & $S(1) \sim \text{Bernoulli}\{\mu_1^S(X)\}$ with $\mu_1^S(X) = \text{expit}(0.1X_1 + 0.1X_2)$ \\
%    &  $S(0) \sim \text{Bernoulli}\{\mu_0^S(X)\}$ with $\mu_0^S(X) = \text{expit}(0.5X_1 + 0.2X_2)$    
%  \end{align*}
The observed outcomes $Y$ and $S$ are generated accordingly: $Y = A\,Y(1) + \left(1-A\right)Y(0)$ and $S = A\,S(1) + \left(1-A\right)S(0)$.

% We independently generate three datasets: $\mathcal{D}_{1, 1}$ and $D_{1, 2}$ containing observations $(A, X, Y)$, along with $\mathcal{D}_2$ containing observations $(A, X, S)$. Subsequently, We evaluate ITRs across quantile thresholds $\lambda \in \{0.1, 0.2, 0.3, 0.4\}$. For the estimation process, we combine $\mathcal{D}_{1,2}$ and $\mathcal{D}_2$ to form an auxiliary sample used to estimate the nuisance parameters $\hat{\eta}_\lambda$. The main sample $\mathcal{D}_{1,1}$ is reserved exclusively for constructing the final estimators.
Two datasets are randomly generated: $\mathcal{D}_1 = \{(A_i, X_i, Y_i)\}_{i=1}^n$ and $\mathcal{D}_2 = \{(A_j, X_j, S_j)\}_{j=1}^n$. We then evaluate ITRs across quantile thresholds $\lambda \in \{0.1, 0.2, 0.3, 0.4\}$. Following Definition \ref{def:estimator}, we split $\mathcal{D}_1$ into $\mathcal{D}_{1,1}$ and $\mathcal{D}_{1, 2}$ and 
use the combined auxiliary sample $\mathcal{D}_{1,2} \cup \mathcal{D}_2$ to estimate the nuisance parameters $\hat{\eta}$, while reserving $\mathcal{D}_{1,1}$ for constructing the final estimators.
% We employ a sample-splitting strategy on $\mathcal{D}_1$. Specifically, $\mathcal{D}_1$ is randomly partitioned into two disjoint subsets, $\mathcal{D}_{1,1}$ and $\mathcal{D}_{1,2}$ with equal probability. We use the auxiliary sample $\mathcal{D}_{1,2} \cup \mathcal{D}_2$ to estimate the nuisance parameters $\hat{\eta}_\lambda$, and reserve the main sample $\mathcal{D}_{1,1}$ for constructing final estimators.

For implementation, we employ logistic regression to estimate the propensity score $e(X)$, and XGBoost (Extreme Gradient Boosting) \citep{chen2016xgboost} to estimate the outcome regression functions $\mu_a(X)$ and $\mu_{S, a}(X)$ (for $a = 0, 1$), along with the CATEs $\tau_Y(X)$ and $\tau_S(X)$.

% \subsection{Numerical Results}

We replicate each simulation 10,000 times and evaluate performance using bias, standard deviation (SD), and coverage proportion of the 95\% confidence intervals (CP95). Specifically, bias and SD are the Monte Carlo bias and standard deviation of the point estimates of $R$, $R(\lambda)$, $V(\lambda)$, and $G(\lambda)$ across replications. CP95 is the empirical coverage proportion of the 95\% confidence intervals constructed using bootstrap with 5,000 resamples.

Table 1 shows that bias remains small across all sample sizes (500 to 3,000) and quantile thresholds, consistent with the asymptotic consistency of our estimators.
CP95 is close to the nominal 0.95 level across all sample sizes, which is empirically consistent with the theoretical asymptotic normality results in Section \ref{sec:5}.

\subsection{Real Data Analysis: Criteo Dataset}

\begin{table}
    \centering
    \small 
    \setlength{\tabcolsep}{3.5pt} % Default is 6pt
    \begin{tabular}{c|c|ccc|ccc}
    \hline
     & & \multicolumn{3}{c|}{Visit} & \multicolumn{3}{c}{Exposure} \\
    Method & \(\lambda\) & $\hat{R}(\lambda)$ & SD & 95\% CI & $\hat{R}(\lambda)$ & SD & 95\% CI \\
    \hline
    & 5\% & 0.000195 & 0.000035 & [0.000127, 0.000263] & 0.000102 & 0.000033 & [0.000038, 0.000166] \\
    & 10\% & 0.000116 & 0.000032 & [0.000054, 0.000178] & 0.000079 & 0.000029 & [0.000022, 0.000137] \\
    RF & 15\% & 0.000105 & 0.000030 & [0.000046, 0.000164] & 0.000060 & 0.000028 & [0.000006, 0.000115] \\
    & 20\% & 0.000094 & 0.000029 & [0.000037, 0.000151] & 0.000054 & 0.000028 & [0.000000, 0.000108] \\
    \cline{2-8}
    & 5\% & 0.000160 & 0.000033 & [0.000095, 0.000225] & 0.000128 & 0.000033 & [0.000063, 0.000192] \\
    & 10\% & 0.000105 & 0.000026 & [0.000054, 0.000156] & 0.000129 & 0.000028 & [0.000074, 0.000183] \\
    LightGBM & 15\% & 0.000063 & 0.000022 & [0.000020, 0.000105] & 0.000087 & 0.000024 & [0.000041, 0.000134] \\
    & 20\% & 0.000056 & 0.000018 & [0.000021, 0.000090] & 0.000070 & 0.000020 & [0.000030, 0.000110] \\
    \cline{2-8}
    & 5\% & 0.000139 & 0.000035 & [0.000071, 0.000206] & 0.000133 & 0.000036 & [0.000064, 0.000203] \\
    & 10\% & 0.000108 & 0.000028 & [0.000053, 0.000163] & 0.000119 & 0.000031 & [0.000059, 0.000179] \\
    XGBoost & 15\% & 0.000101 & 0.000023 & [0.000056, 0.000145] & 0.000109 & 0.000026 & [0.000059, 0.000159] \\
    & 20\% & 0.000053 & 0.000017 & [0.000020, 0.000087] & 0.000071 & 0.000022 & [0.000028, 0.000114] \\
    \cline{2-8}
    & 5\% & 0.000148 & 0.000038 & [0.000074, 0.000222] & 0.000128 & 0.000035 & [0.000059, 0.000197] \\
    & 10\% & 0.000083 & 0.000032 & [0.000021, 0.000145] & 0.000090 & 0.000031 & [0.000030, 0.000150] \\
    Boosting & 15\% & 0.000086 & 0.000028 & [0.000030, 0.000141] & 0.000049 & 0.000027 & [-0.000003, 0.000102] \\
    & 20\% & 0.000094 & 0.000025 & [0.000044, 0.000144] & 0.000063 & 0.000023 & [0.000017, 0.000108] \\
    \hline
    \end{tabular}
    \caption{Comparison of different machine learning methods for estimating $R(\lambda)$ with visit and exposure as surrogate endpoints across various budget values.}
    \label{tab:rn_comparison}
    \end{table}
    
    \begin{table}
    \centering
    \small  
    \setlength{\tabcolsep}{3.5pt} % Default is 6pt
    \begin{tabular}{c|c|ccc|ccc}
    \hline
     & & \multicolumn{3}{c|}{Visit} & \multicolumn{3}{c}{Exposure} \\
    Method & \(\lambda\) & $\hat{G}(\lambda)$ & SD & 95\% CI & $\hat{G}(\lambda)$ & SD & 95\% CI \\
    \hline
    & 5\% & 0.000421 & 0.000039 & [0.000344, 0.000497] & 0.000506 & 0.000042 & [0.000424, 0.000588] \\
    & 10\% & 0.000547 & 0.000046 & [0.000458, 0.000637] & 0.000592 & 0.000048 & [0.000498, 0.000687] \\
    RF & 15\% & 0.000607 & 0.000048 & [0.000512, 0.000702] & 0.000651 & 0.000052 & [0.000550, 0.000752] \\
    & 20\% & 0.000650 & 0.000050 & [0.000553, 0.000748] & 0.000679 & 0.000054 & [0.000574, 0.000785] \\
    \cline{2-8}
    & 5\% & 0.000497 & 0.000042 & [0.000415, 0.000579] & 0.000531 & 0.000041 & [0.000450, 0.000612] \\
    & 10\% & 0.000632 & 0.000050 & [0.000535, 0.000729] & 0.000608 & 0.000048 & [0.000513, 0.000702] \\
    LightGBM & 15\% & 0.000700 & 0.000053 & [0.000595, 0.000804] & 0.000675 & 0.000052 & [0.000574, 0.000776] \\
    & 20\% & 0.000730 & 0.000055 & [0.000622, 0.000839] & 0.000716 & 0.000054 & [0.000610, 0.000822] \\
    \cline{2-8}
    & 5\% & 0.000301 & 0.000043 & [0.000216, 0.000386] & 0.000306 & 0.000041 & [0.000225, 0.000387] \\
    & 10\% & 0.000374 & 0.000050 & [0.000277, 0.000472] & 0.000364 & 0.000048 & [0.000271, 0.000457] \\
    XGBoost & 15\% & 0.000408 & 0.000053 & [0.000303, 0.000512] & 0.000399 & 0.000051 & [0.000299, 0.000500] \\
    & 20\% & 0.000470 & 0.000056 & [0.000360, 0.000579] & 0.000452 & 0.000054 & [0.000347, 0.000558] \\
    \cline{2-8}
    & 5\% & 0.000456 & 0.000038 & [0.000382, 0.000530] & 0.000417 & 0.000076 & [0.000269, 0.000565] \\
    & 10\% & 0.000530 & 0.000047 & [0.000438, 0.000621] & 0.000461 & 0.000089 & [0.000286, 0.000635] \\
    Boosting & 15\% & 0.000523 & 0.000051 & [0.000423, 0.000623] & 0.000482 & 0.000096 & [0.000294, 0.000670] \\
    & 20\% & 0.000519 & 0.000053 & [0.000415, 0.000623] & 0.000479 & 0.000101 & [0.000281, 0.000677] \\
    \hline
    \end{tabular}
    \caption{Comparison of different machine learning methods for estimating $G(\lambda)$ with visit and exposure as surrogate endpoints across various budget values.}
    \label{tab:gn_comparison}
    \end{table}
    
    \begin{table}
        \centering
        \small  
        \setlength{\tabcolsep}{3.5pt} % Default is 6pt
        \begin{tabular}{c|c|ccc|ccc}
        \hline
         & & \multicolumn{3}{c|}{Visit} & \multicolumn{3}{c}{Exposure} \\
        Method & \(\lambda\) & $\hat{V}(\lambda)$ & SD & 95\% CI & $\hat{V}(\lambda)$ & SD & 95\% CI \\
        \hline
        & 5\% & 0.000377 & 0.000037 & [0.000304, 0.000450] & 0.000469 & 0.000040 & [0.000391, 0.000547] \\
        & 10\% & 0.000473 & 0.000041 & [0.000391, 0.000554] & 0.000510 & 0.000044 & [0.000424, 0.000595] \\
        RF & 15\% & 0.000491 & 0.000042 & [0.000409, 0.000572] & 0.000535 & 0.000044 & [0.000449, 0.000622] \\
        & 20\% & 0.000481 & 0.000040 & [0.000402, 0.000560] & 0.000521 & 0.000043 & [0.000436, 0.000606] \\
        \cline{2-8}
        & 5\% & 0.000452 & 0.000040 & [0.000374, 0.000530] & 0.000486 & 0.000039 & [0.000408, 0.000563] \\
        & 10\% & 0.000541 & 0.000045 & [0.000453, 0.000629] & 0.000517 & 0.000044 & [0.000432, 0.000602] \\
        LightGBM & 15\% & 0.000564 & 0.000045 & [0.000475, 0.000653] & 0.000539 & 0.000044 & [0.000452, 0.000626] \\
        & 20\% & 0.000549 & 0.000044 & [0.000462, 0.000636] & 0.000535 & 0.000044 & [0.000450, 0.000620] \\
        \cline{2-8}
        & 5\% & 0.000270 & 0.000041 & [0.000189, 0.000351] & 0.000275 & 0.000039 & [0.000198, 0.000352] \\
        & 10\% & 0.000312 & 0.000045 & [0.000225, 0.000400] & 0.000302 & 0.000043 & [0.000218, 0.000386] \\
        XGBoost & 15\% & 0.000315 & 0.000045 & [0.000226, 0.000403] & 0.000306 & 0.000044 & [0.000220, 0.000392] \\
        & 20\% & 0.000346 & 0.000045 & [0.000258, 0.000433] & 0.000328 & 0.000043 & [0.000243, 0.000413] \\
        \cline{2-8}
        & 5\% & 0.000419 & 0.000036 & [0.000348, 0.000490] & 0.000439 & 0.000039 & [0.000363, 0.000515] \\
        & 10\% & 0.000456 & 0.000042 & [0.000373, 0.000539] & 0.000449 & 0.000043 & [0.000365, 0.000533] \\
        Boosting & 15\% & 0.000412 & 0.000043 & [0.000327, 0.000497] & 0.000448 & 0.000044 & [0.000362, 0.000534] \\
        & 20\% & 0.000371 & 0.000043 & [0.000288, 0.000455] & 0.000403 & 0.000043 & [0.000318, 0.000488] \\
        \hline
        \end{tabular}
        \caption{Comparison of different machine learning methods for estimating $V(\lambda)$ with visit and exposure as surrogate endpoints across various budget values.}
        \label{tab:vn_comparison}
        \end{table}

To further evaluate our approach in real-world applications, we analyze the public Criteo Uplift Prediction dataset on digital advertising \citep{Diemert2018}. This dataset aggregates multiple incrementality experiments in which a subset of users was randomly withheld from receiving advertisements. It contains 25,309,483 user-level observations, each with a binary treatment indicator, 12 covariates, two surrogate endpoints (visit and exposure) and a conversion indicator. 
% In this part, we use experimental data from the Criteo Uplift Prediction Dataset. We employ the dataset analyzed in \citep{Diemert2018}, which was constructed by assembling data from several incrementality tests, where a random portion of the population was prevented from being targeted by advertising. The dataset comprises 25,309,483 observations, each representing a user with a binary treatment variable, 12 features, and two surrogate labels: visit and exposure. 
The data exhibit a visit rate of 4.13\%, a conversion rate of 0.23\%, and a treatment ratio of 84.6\%, reflecting the high sparsity typical in digital advertising conversion. 
We apply our estimators to evaluate the utility of using Visit and Exposure as surrogates in ITR design.
% We apply our proposed estimator to evaluate using Visit and Exposure as surrogates for ITR estimation.

In our experimental design, the primary outcome is conversion (whether a user made a purchase), while the surrogates include visit (whether a user visited the website) and exposure (whether the user was effectively exposed to the treatment). The treatment variable $A$ indicates whether a user was assigned to receive the advertisement. Additionally, all 12 available covariate features are incorporated into the analysis.
We then apply the same sample-splitting strategy as in the simulation study, partitioning the dataset into auxiliary and main samples.

We apply four machine learning methods to estimate CATEs: Random Forest (RF)\citep{breiman2001random}, LightGBM \citep{ke2017lightgbm}, XGBoost \citep{pedregosa2011scikit}, and Gradient Boosting \citep{natekin2013gradient}. For treatment assignment probability, we employ a logistic regression model to estimate the propensity scores. 
% Logistic regression is implemented through sklearn's LogisticRegression, RandomForestRegressor is used for random forests, GradientBoostingRegressor for boosting, lightgbm package for lightGBM, and ``xgboost'' package for XGBoost.
Implementations use \texttt{scikit-learn}'s \texttt{LogisticRegression} for the propensity model, \texttt{RandomForestRegressor} for random forests, \texttt{GradientBoostingRegressor} for gradient boosting, \texttt{lightgbm} for LightGBM, and \texttt{xgboost} for XGBoost.

We evaluate ITR performance under budget constraints $\lambda \in \{5\%, 10\%, 15\%, 20\%\}$, where $\lambda$ represents the maximum proportion of the population that can receive treatment. For each configuration, we estimate the three proposed measures with corresponding SD and $95\%$ confidence intervals: (i) the $\lambda$-surrogate regret $R(\lambda)$, quantifiying the performance gap between the $\lambda$-surrogate-optimal ITR and the $\lambda$-outcome-optimal ITR (smaller is better); (ii) the $\lambda$-surrogate gain $G(\lambda)$, measuring the net benefit of the $\lambda$-surrogate-optimal ITR relative to a no-treatment baseline (larger is better), and (iii) the $\lambda$-surrogate efficiency $V(\lambda)$, comparing the expected outcome of the $\lambda$-surrogate-optimal ITR with random treatment allocation (larger is better).

Tables 1-3 present results across all machine learning methods, surrogate endpoints, and budget constraints. 
Surrogate regret remains small in all settings (ranging from 0.000049 to 0.000195), with surrogate-optimal ITRs closely matching the performance of outcome-optimal ITRs. 
Furthermore, the consistently positive surrogate efficiency indicates that Visit and Exposure offer substantial information gain for decision-making in advertising.
Finally, performance generally improves with increased budgets, as regret decreases while gain and efficiency rise, suggesting that larger treatment budgets enable the $\lambda$-surrogate-optimal ITR to more closely approximate the $\lambda$-outcome-optimal ITR.

\section{Discussions}
In this work, we develop a framework for evaluating surrogate endpoints in the ITR setting. This framework introduces three complementary evaluation measures: surrogate regret, surrogate gain, and surrogate efficiency, along with their corresponding estimators and asymptotic properties. These measures 
provides guidance on when surrogates can aid ITR decision-making that targets the primary outcome.

Several promising directions for future research remain. First, methods that combine multiple surrogates into an optimal composite surrogate would be 
valuable. Recent work by \citet{athey2025surrogate} has explored surrogate index construction, extending such approaches to our framework is an interesting direction. 
% Second, when policy class restrictions are imposed, understanding how to find the best policy within the restricted class and establish the corresponding estimation bounds remains important. \citet{Kitagawa2018} studied the gap between empirical and oracle policies for treatment effects, and similar analysis in our setting would provide valuable insights. 
Second, extending our framework to handle missing data scenarios, where either surrogate or outcome measurements are partially unavailable, would broaden its applicability. \citet{kallus2025role} have developed methods for estimating ITRs with missing outcomes, and building upon such techniques is a natural next step.

Overall, the theory developed in this work contributes to understanding surrogate endpoints in ITR estimation and the trade-offs between surrogate and primary objectives. Rigorous surrogate evaluation can improve decision-making in applications where parimary outcomes are costly or delayed, and our framework provides practical tools for practitioners.

\bibliographystyle{plainnat}
\bibliography{reference}

% \section*{Data availability statement}

%This paper does not contain data analysis. 

%\section*{Figures}
%
%Figure \ref{fig}: Violin plot of the differences between 2 times the \olss coefficients of $(A_i, B_i, A_iB_i)$ and the true values of $(\ta, \tb 	, \tab)$  over $100,000$ independent complete randomizations. We use ``N'', ``F'', and ``L'' to represent the factor-saturated unadjusted, additive, and fully interacted regressions, respectively, and use ``N\_us'', ``F\_us'', and ``L\_us'' to  represent their respective factor-unsaturated variants. See Table \ref{tb:sim}.

 \newpage 
 \appendix
 
\begin{center}
\bf \Large 
Supplementary Material  for ``Evaluating  \\ Surrogates in Individualized Treatment Rules"
\end{center}

\setcounter{equation}{0}
\setcounter{section}{0}
\setcounter{figure}{0}
\setcounter{example}{0}
\setcounter{proposition}{0}
\setcounter{corollary}{0}
\setcounter{theorem}{0}
\setcounter{table}{0}
\setcounter{condition}{0}
\setcounter{lemma}{0}
\setcounter{remark}{0}

\renewcommand {\theproposition} {S\arabic{proposition}}
\renewcommand {\theexample} {S\arabic{example}}
\renewcommand {\thefigure} {S\arabic{figure}}
\renewcommand {\thetable} {S\arabic{table}}
\renewcommand {\theequation} {S\arabic{equation}}
\renewcommand {\thelemma} {S\arabic{lemma}}
\renewcommand {\thesection} {S\arabic{section}}
\renewcommand {\thetheorem} {S\arabic{theorem}}
\renewcommand {\thecorollary} {S\arabic{corollary}}
\renewcommand {\thecondition} {S\arabic{condition}}
\renewcommand {\thepage} {S\arabic{page}}
\renewcommand {\theremark} {S\arabic{remark}}

\setcounter{page}{1}

  \setcounter{equation}{0}
\renewcommand {\theequation} {S\arabic{equation}}
  \setcounter{lemma}{0}
\renewcommand {\thelemma} {S\arabic{lemma}}
   \setcounter{definition}{0}
\renewcommand {\thedefinition} {S\arabic{definition}}
   \setcounter{example}{0}
\renewcommand {\theexample} {S\arabic{example}}
   \setcounter{proposition}{0}
\renewcommand {\theproposition} {S\arabic{proposition}}
   \setcounter{corollary}{0}
\renewcommand {\thecorollary} {S\arabic{corollary}}

 % \bigskip 

\section{The optimal transformation framework in ITR}
\label{appendix 1}

The goal of the optimal transformation framework by \citet{wang2020model} is to find an optimal function of $S$, $g(\cdot)$, such that $g(S)$ can be used to approximate the primary outcome and subsequently to quantify the treatment effect on $Y$. Mathematically,
it aims to identify $g(\cdot)$ that minimizes the following mean squared error loss function:
$$L(g)=E\left[\left\{ \left(Y(1)-Y(0)\right)-\left[g(S(1))-g(S(0))\right]\right\}^2\right].$$

Since the potential outcomes $Y(1), Y(0), S(1)$, and $S(0)$ are never jointly observed for the same individual, the authors proposed minimizing an alternative observable function:

\begin{equation}
\label{transform}
\min\quad  E\left\{\left[Y-g(S)\right]^2 \right\} \quad \text{s.t.} \quad E\left\{\left[Y-g(S)\right] \mid A=0\right\}=0.
\end{equation}
% \label{eq:1}

Let's consider the following numerical example:

\begin{example}
Suppose that $X$ takes values in $\{-1,0,1\}$, each with probability $1/3$.  The joint distribution of the potential outcomes $(S,Y)$ is given by:

$$
(S(0),S(1),Y(0),Y(1)) = 
\begin{cases} 
(2,3,4,3), & X=-1, \\  % First case: expression and condition
(3,2,3,4), & X=0,    \\    % Second case: trailing comma optional
(2,2,0,0), &   X=1.
\end{cases}
$$
\end{example}

Applying their method to this example, we solve the optimization problem in \eqref{transform} and obtain the optimal transformation $g(s)=s$. We denote the optimal ITR derived based on this transformed surrogate $g(S)$ as the transformed-surrogate-optimal ITR. Under this rule, individuals with $X=-1$ are assigned to the treatment group. The resulting expected outcome under this ITR is $2$.
In contrast, the outcome-optimal ITR treats individuals with $X=0$, achieving a higher value of $8/3$. 
For any budget level $\lambda$, consider the randomized treatment rule $\pi_\lambda$ that assigns treatment according to
$\text{Bernoulli}(\lambda)$. The expected outcome under $\pi_\lambda$ is 
$7/3$. Therefore, the transformed-surrogate-optimal ITR performs worse than both the outcome-optimal ITR and $\pi_\lambda$ for all $\lambda$.

\section{Estimation Algorithms}
\label{appendix3}
In this section, we present three supplementary estimation algorithms that complement the methods discussed in the main text. Algorithm S1 presents the procedure for the split data case corresponding to discussions in Section 5.1, while Algorithms S2 and S3 outline the approach for the single dataset case in Section 2.1.

These algorithms are designed to estimate a general parameter $\theta$, where $\theta$ represents a specific measure of interest such as the surrogate regret $R$, the $\lambda$-surrogate regret $R(\lambda)$, the $\lambda$-surrogate gain $G(\lambda)$, or the $\lambda$-surrogate efficiency $V(\lambda)$, depending on the specific form of the estimating function $\phi$.
% \vspace{0.5em}
% \hrule
% \vspace{0.5em}

\subsection{Cross-fitting for Split Dataset Case}
\label{appendix: B1}

We now present the algorithm for the two-sample setting discussed in Section 5.1, where we have access to two separate datasets: $\mathcal{D}_1$ containing outcomes $(A, X, Y)$ but no surrogates, and $\mathcal{D}_2$ containing surrogates $(A, X, S)$ but no outcomes. Algorithm S1 below describes the construction of the estimator using cross-fitting on $\mathcal{D}_1$.

\noindent\rule[0.25\baselineskip]{\textwidth}{1pt}
{\bf Algorithm S1.}
\label{alg:cross_fitting}

\textbf{Data Setup:} Partition the outcome dataset index set $I_1$ into $K$ disjoint folds $I_{(1)}, \ldots, I_{(K)}$ of equal size. Let $I_{(k)}^C = I_1 \setminus I_{(k)}$ denote the complement of $I_{(k)}$ for $k = 1, \ldots, K$. Let $I_2$ denote the index set of the surrogate dataset $\mathcal{D}_2$.

\textbf{Step 1: Nuisance parameter training with cross-fitting.}

Using $\mathcal{D}_2$ (indices $I_2$), estimate propensity score $\hat{e}$ and surrogate regression functions $\hat{\mu}_{S,0}, \hat{\mu}_{S,1}$.

\quad \textbf{for} $k = 1$ to $K$ \textbf{do}

    \quad (1) Using fold $I_{(k)}^C$, estimate outcome regression functions $\hat{\mu}_0, \hat{\mu}_1$.
    
    \quad (2) \textit{CATEs:} Compute $\hat{\tau}_Y(x) = \hat{\mu}_1(x) - \hat{\mu}_0(x)$ and $\hat{\tau}_S(x) = \hat{\mu}_{S,1}(x) - \hat{\mu}_{S,0}(x)$.
    
    \quad (3) \textit{Thresholds and ITRs:}

    \quad \textbf{Case 1: Unconstrained setting ($\lambda = 1$).}
    \quad Thresholds are not required; we have 
    $$\hat{\pi}_{Y}(x) = \mathbf{1}\{\hat{\tau}_Y(x) > 0\}, \hat{\pi}_{S}(x) = \mathbf{1}\{\hat{\tau}_S(x) > 0\}.$$
    
    \quad \textbf{Case 2: Budget-constrained setting ($\lambda < 1$).}
    \quad Using $I_2$ and $I_{(k)}^C$ respectively, we estimate the quantiles $\hat{y}_{1-\lambda}$ and $\hat{s}_{1-\lambda}$:
    \[
        \hat{y}_{1-\lambda} = \inf\left\{t \in \mathbb{R} : \frac{\sum_{j \in I_2} \mathbf{1}\{\hat{\tau}_Y(X_j) \leq t\}}{|I_2|} \geq 1-\lambda\right\},
    \]
    \[
        \hat{s}_{1-\lambda} = \inf\left\{t \in \mathbb{R} : \frac{\sum_{j \in I_{(k)}^C} \mathbf{1}\{\hat{\tau}_S(X_j) \leq t\}}{|I_{(k)}^C|} \geq 1-\lambda\right\}.
    \]
    \quad The estimated ITR is
    $$\hat{\pi}_{Y,\lambda}(x) = \mathbf{1}\{\hat{\tau}_Y(x) > \hat{y}_{1-\lambda}\} \cdot \mathbf{1}\{\hat{\tau}_Y(x) > 0\}, 
    \hat{\pi}_{S,\lambda}(X) = \mathbf{1}\{\hat{\tau}_S(X) > \hat{s}_{1-\lambda}\}\cdot\mathbf{1}\{\hat{\tau}_S(X) > 0\}.
    $$    
    
    \quad (4) Construct estimates 
    \[
    \hat{\eta}_{\lambda, -k} = 
    \begin{cases} 
    \left( \hat{\mu}_0, \hat{\mu}_1, \hat{\mu}_{S,0}, \hat{\mu}_{S,1}, \hat{e}, \hat{\pi}_{Y, \lambda}, \hat{\pi}_{S, \lambda}, \hat{y}_{1-\lambda}, \hat{s}_{1-\lambda}\right), &\lambda < 1, \\
    \left( \hat{\mu}_0, \hat{\mu}_1, \hat{\mu}_{S,0}, \hat{\mu}_{S,1}, \hat{e}, \hat{\pi}_{Y}, \hat{\pi}_{S}\right), & \lambda = 1.
    \end{cases}
    \]
    For notational simplicity, we suppress the dependence on $\lambda$ and denote the nuisance parameter estimator as $\hat{\eta}_{-k}$ hereafter.
    % $\hat{\eta}_{\lambda, -k} = \left( \hat{\mu}_0, \hat{\mu}_1, \hat{\mu}_{S,0}, \hat{\mu}_{S,1}, \hat{e}, \hat{\pi}_{Y, \lambda}, \hat{\pi}_{S, \lambda}, \hat{y}_{1-\lambda}, \hat{s}_{1-\lambda}\right)$

    \quad (5) Obtain the predicted values of $\eta(X_i)$ for $i \in I_{(k)}$, denoted as $\hat{\eta}_{-k}(X_i)$.
    
\quad \textbf{end}

% All the predicted values of $\hat{\eta}_{\lambda, -k}(X_i)$ for $i \in I_1$ consist of the final estimates of $\eta(X_i)$, denoted as $\hat{\eta}(X_i)$.

% \vspace{0.3em}
\textbf{Step 2: Constructing the final estimator.}
The proposed estimator of $\hat{\theta}$ is given as
\begin{equation*}
    \hat{\theta} = K^{-1}\sum_{k=1}^K
    \left(|I_{(k)}|^{-1}\sum_{i \in I_{(k)}}
    \phi(Y_i, A_i, X_i; \hat{\eta}_{-k})\right).
\end{equation*}

\noindent\rule[0.25\baselineskip]{\textwidth}{1pt}

\begin{remark}
    The cross-fitting estimator shares the same asymptotic properties as the sample-splitting estimator. We illustrate this equivalence by revisiting the error decomposition used in the proof of Theorem 1 (see Section \ref{pf:1}). 
    Let $\hat{R}^{(k)}$ denote the estimator computed on the $k$-th fold $I_{(k)}$ using nuisance parameters $\hat{\eta}_{\lambda, -k}$ estimated from the complementary folds. Let $E_{I_{(k)}}$ denote the empirical expectation over the set $I_{(k)}$. The estimation error for a single fold can be decomposed as:
    \begin{align}
        \hat{R}^{(k)} - R &= 
        E_{I_{(k)}}\left[\phi_R(Y, A, X;\, \hat{\eta}_{-k})\right] - E_P\left[\phi_R(Y, A, X;\, \eta)\right]\notag
        \\
        &= 
        \underbrace{(E_{I_{(k)}}-E_P)\left[\phi_R(Y, A, X;\, \eta)\right]}_{\text{Leading Term}}
        \nonumber\\
        &\quad + \underbrace{(E_{I_{(k)}}-E_P)\left[\phi_R(Y, A, X;\, \hat{\eta}_{-k}) - \phi_R(Y, A, X;\, \eta)\right]}_{T_1^{(k)}} \nonumber
        \\
        &\quad + \underbrace{E_P\left[\phi_R(Y, A, X;\, \hat{\eta}_{-k}) - \phi_R(Y, A, X;\, \eta)\right]}_{T_2^{(k)}}. \notag
    \end{align}
    The cross-fitting estimator is given by $\hat{R}_{CF} = K^{-1} \sum_{k=1}^K \hat{R}^{(k)}$. We analyze these terms as follows:

\begin{enumerate}
    \item Term $T_1^{(k)}$:
    Conditional on the data split, the nuisance parameters $\hat{\eta}_{-k}$ are independent of the evaluation fold $I_{(k)}$. This independence allows the application of Lemma \ref{lemma: B1} to each fold individually. Consequently, the term $T_1^{(k)}$ is asymptotically negligible ($o_P(n^{-1/2})$) for every $k$. The average of these terms in $\hat{R}_{CF}$ remains $o_P(n^{-1/2})$.
    
    \item Term $T_2^{(k)}$:
    Provided that Assumption 3 holds for the subsamples used to estimate $\hat{\eta}_{-k}$, the bound derived in Theorem 1 applies to each fold $\hat{R}^{(k)}$. Averaging these terms over $K$ folds preserves the convergence rate $O_P(D_{1,n} + D_{2,n} + D_{3,n})$.
    
    \item Asymptotic Normality (Leading Term):
    Provided the condition $D_{1,n} + D_{2,n} + D_{3,n} = o_P(n^{-1/2})$ holds and the bias terms are negligible, the asymptotic distribution is determined by the average of the leading terms. Since the folds form a partition of the full dataset $I_1$ (assuming equal fold sizes for simplicity), the average of empirical expectations over folds is equivalent to the empirical expectation over the full set $E_{I_1}$:
    \[
    \frac{1}{K} \sum_{k=1}^K (E_{I_{(k)}}-E_P)\left[\phi_R(\cdot;\, \eta)\right] = (E_{I_1}-E_P)\left[\phi_R(\cdot;\, \eta)\right].
    \]
    Therefore, the cross-fitting estimator $\hat{R}_{CF}$ satisfies:
    \[
    |I_1|^{1/2}\left(\hat{R}_{CF} - R\right) = |I_1|^{1/2} (E_{I_1}-E_P)\left[\phi_R(Y, A, X;\, \eta)\right] + o_P(1).
    \]
    Thus, $\hat{R}_{CF}$ achieves the same asymptotic variance $\sigma_R^2$ as the sample-splitting estimator stated in Proposition 1.
\end{enumerate}
\end{remark}

\subsection{Sample Splitting for Single Dataset Case}
\label{appendix: B2}
We now present the sample splitting algorithm for the single dataset case introduced in Section 2.1, where we have access to a unified dataset $\mathcal{D}$ containing both outcomes and surrogates $(A, X, S, Y)$ for all observations.

\noindent\rule[0.25\baselineskip]{\textwidth}{1pt}
{\bf Algorithm S2.}
\label{alg:sample_splitting_simple}

\textbf{Data Setup:} 
Given the single dataset $\mathcal{D} = \{(A_i, X_i, S_i, Y_i)\}_{i=1}^m$, we randomly partition the index set $\{1, \dots, m\}$ into two disjoint sets of equal size: $I_1$ and $I_2$.

\textbf{Step 1: Nuisance parameter training on $I_2$.}

    \quad (1) Randomly partition the indices $I_2$ into two disjoint halves $J_{1}$ and $J_{2}$ of equal size.
    
    \quad (2) Using the full $I_2$, estimate the propensity score $\hat{e}$.
    
    \quad (3) Using fold $J_{1}$, estimate the outcome regression functions $\{\hat{\mu}_0, \hat{\mu}_1\}$; using fold $J_{2}$, estimate the surrogate regression functions $\{\hat{\mu}_{S, 0}, \hat{\mu}_{S, 1}\}$.
    
    \quad (4) \textit{CATEs:} Compute $\hat{\tau}_Y(x) = \hat{\mu}_1(x) - \hat{\mu}_0(x)$ and $\hat{\tau}_S(x) = \hat{\mu}_{S,1}(x) - \hat{\mu}_{S,0}(x)$.
    
    \quad (5) \textit{Thresholds and ITRs:}

    \quad \textbf{Case 1: Unconstrained setting ($\lambda = 1$).}
    \quad Thresholds are not required; we have 
    $$\hat{\pi}_{Y}(x) = \mathbf{1}\{\hat{\tau}_Y(x) > 0\}, \hat{\pi}_{S}(x) = \mathbf{1}\{\hat{\tau}_S(x) > 0\}.$$
    
    \quad \textbf{Case 2: Budget-constrained setting ($\lambda < 1$).}
    \quad Utilizing $J_{k,1}$ and $J_{k,2}$ respectively, we estimate the quantiles $y_{1-\lambda}$, $s_{1-\lambda}$:
    \[
        \hat{y}_{1-\lambda} = \inf\left\{t \in \mathbb{R} : 
        \frac{\sum_{j \in J_{2}} \mathbf{1}\{\hat{\tau}_Y(X_j) \leq t\}}{|J_{2}|} \geq 1-\lambda\right\},
    \]
    \[
        \hat{s}_{1-\lambda} = \inf\left\{t \in \mathbb{R} : 
        \frac{\sum_{j \in J_{1}} \mathbf{1}\{\hat{\tau}_S(X_j) \leq t\}}{|J_{1}|} \geq 1-\lambda\right\}.
    \]
    \quad The estimated ITRs are
    $$\hat{\pi}_{Y,\lambda}(x) = \mathbf{1}\{\hat{\tau}_Y(x) > \hat{y}_{1-\lambda}\} \cdot \mathbf{1}\{\hat{\tau}_Y(x) > 0\}, 
    \hat{\pi}_{S,\lambda}(X) = \mathbf{1}\{\hat{\tau}_S(X) > \hat{s}_{1-\lambda}\}\cdot\mathbf{1}\{\hat{\tau}_S(X) > 0\}.
    $$    
    \quad (6) Construct estimates \[
    \hat{\eta} = 
    \begin{cases} 
    \left( \hat{\mu}_0, \hat{\mu}_1, \hat{\mu}_{S,0}, \hat{\mu}_{S,1}, \hat{e}, \hat{\pi}_{Y, \lambda}, \hat{\pi}_{S, \lambda}, \hat{y}_{1-\lambda}, \hat{s}_{1-\lambda}\right), &\lambda < 1, \\
    \left( \hat{\mu}_0, \hat{\mu}_1, \hat{\mu}_{S,0}, \hat{\mu}_{S,1}, \hat{e}, \hat{\pi}_{Y}, \hat{\pi}_{S}\right), & \lambda = 1.
    \end{cases}
    \]

\textbf{Step 2: Constructing the final estimator.}
The proposed estimator of \(\theta\), denoted as $\hat{\theta}$, is given by:
\begin{equation*}
    \hat{\theta} = |I_{1}|^{-1}\sum_{i \in I_{1}}
    \phi(Y_i, A_i, X_i; \hat{\eta}).
\end{equation*}

\noindent\rule[0.25\baselineskip]{\textwidth}{1pt}

\subsection{Cross-fitting for Single Dataset Case}
To improve finite-sample performance, we now present the cross-fitting algorithm. 
% \newpage
% \hrule
% \vspace{0.5em}

\noindent\rule[0.25\baselineskip]{\textwidth}{1pt}
{\bf Algorithm S3.}
\label{alg:sample_splitting_unified}

\textbf{Data Setup:} 
In the scenario where only a single dataset $\mathcal{D} = \{(A_i, X_i, S_i, Y_i)\}_{i=1}^m$ is available, we randomly partition the index set $\{1, \dots, m\}$ into $K$ disjoint folds $I_{(1)}, \ldots, I_{(K)}$ of equal size. Let $I_{(k)}^C = \{1, \dots, m\} \setminus I_{(k)}$ denote the complement of $I_{(k)}$ (the training set) for $k = 1, \ldots, K$.

\textbf{Step 1: Nuisance parameter training with cross-fitting.}

\quad \textbf{for} $k = 1$ to $K$ \textbf{do}

    \quad (1) Randomly partition the training indices $I_{(k)}^C$ into two disjoint halves $J_{k,1}$ and $J_{k,2}$ of equal size.
    
    \quad (2) Using the full training set $I_{(k)}^C$, estimate the propensity score $\hat{e}$.
    
    \quad (3) Using folds $J_{k,1}$ and $J_{k,2}$, estimate the outcome regression functions $\{\hat{\mu}_0, \hat{\mu}_1\}$ and the surrogate regression functions $\{\hat{\mu}_{S, 0}, \hat{\mu}_{S, 1}\}$, respectively.
    
    \quad (4) \textit{CATEs:} Compute $\hat{\tau}_Y(x) = \hat{\mu}_1(x) - \hat{\mu}_0(x)$ and $\hat{\tau}_S(x) = \hat{\mu}_{S,1}(x) - \hat{\mu}_{S,0}(x)$.
    
    \quad (5) \textit{Thresholds and ITRs:}

    \quad \textbf{Case 1: Unconstrained setting ($\lambda = 1$).}
    \quad The thresholds are not required, we have 
    $$\hat{\pi}_{Y}(x) = \mathbf{1}\{\hat{\tau}_Y(x) > 0\}, \hat{\pi}_{S}(x) = \mathbf{1}\{\hat{\tau}_S(x) > 0\}.$$
    
    \quad \textbf{Case 2: Budget-constrained setting ($\lambda < 1$).}
    \quad Utilizing $J_{k,1}$ and $J_{k,2}$ respectively, we estimate the quantiles $\hat{y}_{1-\lambda}$ and $\hat{s}_{1-\lambda}$:
    \[
        \hat{y}_{1-\lambda} = \inf\left\{t \in \mathbb{R} : 
        \frac{\sum_{j \in J_{k,2}} \mathbf{1}\{\hat{\tau}_Y(X_j) \leq t\}}{|J_{k,2}|} \geq 1-\lambda\right\},
    \]
    \[
        \hat{s}_{1-\lambda} = \inf\left\{t \in \mathbb{R} : 
        \frac{\sum_{j \in J_{k,1}} \mathbf{1}\{\hat{\tau}_S(X_j) \leq t\}}{|J_{k,1}|} \geq 1-\lambda\right\}.
    \]
    \quad The estimated ITR is
    $$\hat{\pi}_{Y,\lambda}(x) = \mathbf{1}\{\hat{\tau}_Y(x) > \hat{y}_{1-\lambda}\} \cdot \mathbf{1}\{\hat{\tau}_Y(x) > 0\}, 
    \hat{\pi}_{S,\lambda}(X) = \mathbf{1}\{\hat{\tau}_S(X) > \hat{s}_{1-\lambda}\}\cdot\mathbf{1}\{\hat{\tau}_S(X) > 0\}.
    $$    
    \quad (6) Construct estimates \[
    \hat{\eta}_{-k} = 
    \begin{cases} 
    \left( \hat{\mu}_0, \hat{\mu}_1, \hat{\mu}_{S,0}, \hat{\mu}_{S,1}, \hat{e}, \hat{\pi}_{Y, \lambda}, \hat{\pi}_{S, \lambda}, \hat{y}_{1-\lambda}, \hat{s}_{1-\lambda}\right), &\lambda < 1, \\
    \left( \hat{\mu}_0, \hat{\mu}_1, \hat{\mu}_{S,0}, \hat{\mu}_{S,1}, \hat{e}, \hat{\pi}_{Y}, \hat{\pi}_{S}\right), & \lambda = 1.
    \end{cases}
    \]
    
\quad \textbf{end}

\textbf{Step 2: Constructing the final estimator.}
The proposed estimator of $\hat{\theta}$ is given as
\begin{equation*}
    \hat{\theta} = K^{-1}\sum_{k=1}^K
    \left(|I_{(k)}|^{-1}\sum_{i \in I_{(k)}}
    \phi(Y_i, A_i, X_i; \hat{\eta}_{-k})\right).
\end{equation*}

\noindent\rule[0.25\baselineskip]{\textwidth}{1pt}

\section{Auxiliary Lemmas}
\label{appendix: C}
In this section, we first provide the proofs of the lemmas stated in the main text, followed by several auxiliary results used in the proofs of our main theorems. To simplify the notation in the proofs that follow, we will omit the argument $X$ from functions such as $e(X)$, $\mu_1(X)$, $\tau_Y(X)$, and $\pi_Y(X)$ when the context is clear, denoting them simply as $e$, $\mu_1$, $\pi_Y$, and $\tau_Y$.

\subsection{Proof of Lemma 1}
\begin{proof}
The difference in value between the outcome-optimal ITR and the surrogate-optimal ITR is given by:
\begin{align}
E\left[Y(\pi_Y(X))\right] - E\left[Y(\pi_S(X))\right] & = 
E \left[(\pi_Y - \pi_S) \cdot Y(1) + (1-\pi_Y - 1 + \pi_S) \cdot Y(0)\right]
\notag
\\
&= E \left\{E\left[(\pi_Y - \pi_S) \cdot \left[Y(1)-Y(0)\right] \mid X\right]\right\}
\notag
\\
& = E \left[ \tau_Y \cdot (\pi_Y - \pi_S)\right].
\label{eq:lem1.1}
\end{align}

Since $\pi_Y, \pi_S \in \{0, 1\}$, the term $(\pi_Y - \pi_S)$ is non-zero only when the decisions disagree. We can decompose this into two disjoint cases:

\begin{enumerate}
    \item Under-treatment: $\pi_Y=1$ but $\pi_S=0$ (i.e., $\tau_Y > 0$ and $\tau_S \le 0$). Here, $\pi_Y - \pi_S = 1$.
    \item Over-treatment: $\pi_Y=0$ but $\pi_S=1$ (i.e., $\tau_Y \le 0$ and $\tau_S > 0$). Here, $\pi_Y - \pi_S = -1$.
\end{enumerate}

Substituting these indicators into the expectation:
\begin{align}
    \eqref{eq:lem1.1} &= 
    E \left[ \tau_Y \cdot \mathbf{1}(\tau_Y > 0) \cdot\mathbf{1}(\tau_S \leq 0)\right]
    - E \left[ \tau_Y \cdot \mathbf{1}(\tau_Y \leq 0) \cdot\mathbf{1}(\tau_S > 0)\right]
    \notag
    \\
    &= E \left\{ |\tau_Y| \cdot \left[\mathbf{1}(\tau_Y > 0) \cdot\mathbf{1}(\tau_S \leq 0) + \mathbf{1}(\tau_Y \leq 0) \cdot\mathbf{1}(\tau_S > 0)\right]\right\}. \label{eq:lem1.2}
\end{align}
Since the term inside the expectation in \eqref{eq:lem1.2} is non-negative, the expectation vanishes if and only if the random variable itself is zero almost surely:
\begin{align*}
    \eqref{eq:lem1.2} = 0 
    &\iff |\tau_Y| \cdot \left[\mathbf{1}(\tau_Y > 0) \cdot\mathbf{1}(\tau_S \leq 0) + \mathbf{1}(\tau_Y < 0) \cdot\mathbf{1}(\tau_S > 0)\right] \stackrel{\text{a.s.}}{=} 0 \\
    &\iff \mathbf{1}(\tau_Y > 0) \cdot\mathbf{1}(\tau_S \leq 0) \stackrel{\text{a.s.}}{=} 0 \quad \text{and} \quad \mathbf{1}(\tau_Y < 0) \cdot\mathbf{1}(\tau_S > 0) \stackrel{\text{a.s.}}{=} 0
    \\
    &\iff E \left[ \mathbf{1}(\tau_Y < 0)\cdot\mathbf{1}(\tau_S > 0) \right] = E \left[ \mathbf{1}(\tau_Y > 0)\cdot\mathbf{1}(\tau_S \leq 0) \right] = 0.
\end{align*}
\end{proof}

\subsection{Proof of Lemma 2}
\label{lemma:2}
The proof of Lemma \ref{lemma:2} is a direct consequence of Lemma \ref{lemma: A1} (by setting $\lambda=1$). The algebraic derivations are identical, so we omit the details for brevity. \qed

\subsection{Additional Technical Lemmas and Proofs}

\begin{lemma}[$\lambda$-regret lemma]
\label{lemma: A1}
    For brevity of notation, we denote $\phi_{R, 1}(Y, A, X;\, \eta)$ as $\phi_R(Y, A, X;\, \eta)$, $\pi_{Y, 1}(X)$ as $\pi_Y(X)$, and $\pi_{S, 1}(X)$ as $\pi_S(X)$. Let $\hat{\eta}$ be a collection of estimators for the nuisance parameters $\eta$. Then for any $\lambda \in (0,1]$,
    \begin{align}
        &E\left[\phi_{R, \lambda}(Y, A, X;\, \hat{\eta}) - \phi_{R, \lambda}(Y, A, X;\, \eta)\mid X\right] =
        \notag
        \\
        &   (\hat{\pi}_{Y, \lambda} - \hat{\pi}_{S, \lambda}) \left[ 
        \frac{\hat{e} - e}{\hat{e}}(\hat{\mu}_1 - \mu_1) + 
        \frac{\hat{e} - e}{1 - \hat{e}}(\hat{\mu}_0 - \mu_0)\right] + 
        \tau_Y\left[ (\hat{\pi}_{Y, \lambda} - \pi_{Y, \lambda}) - 
        (\hat{\pi}_{S, \lambda} - \pi_{S, \lambda}) \right]. \notag
    \end{align}
\end{lemma}

\begin{proof}
    Observe that
    \begin{align}
        &E\left[\phi_{R, \lambda}(Y, A, X;\, \hat{\eta}) - \phi_{R, \lambda}(Y, A, X;\, \eta)\mid X\right]\notag
        \\&= E\left\{(\hat{\pi}_{Y, \lambda} - \hat{\pi}_{S, \lambda})\left[\left(\frac{A}{\hat{e}} - 
        \frac{1-A}{1-\hat{e}} \right)(Y - \hat{\mu}_A) + 
       \hat{\tau}_Y\right]\mid X\right\} - 
       \tau_Y(\pi_{Y, \lambda} - \pi_{S, \lambda})\notag
        \\&= (\hat{\pi}_{Y, \lambda} - \hat{\pi}_{S, \lambda})\left\{ 
        E_P\left[ \left(\frac{A}{\hat{e}} - \frac{1-A}{1 - \hat{e}}\right) 
        \left( Y - \hat{\mu}_A \right)  
        \mid X \right] + \hat{\mu}_1 - \hat{\mu}_0 \right\} - 
        \left(\mu_1 - \mu_0\right)\left(\pi_{Y, \lambda} - 
        \pi_{S, \lambda} \right)\notag
        \\
        &= (\hat{\pi}_{Y, \lambda} - \hat{\pi}_{S, \lambda})\left\{ 
        \frac{e}{\hat{e}}(\mu_1 - \hat{\mu}_1) - 
        \frac{1-e}{1 - \hat{e}}(\mu_0 - \hat{\mu}_0) 
         + \hat{\mu}_1 - \hat{\mu}_0 -  \left(\mu_1 - \mu_0\right) \right\} + 
         \left(\mu_1 - \mu_0\right)\left[\hat{\pi}_{Y, \lambda} - \hat{\pi}_{S, \lambda} - (\pi_{Y, \lambda} - 
         \pi_{S, \lambda}) \right]\notag
        \\
        &= (\hat{\pi}_{Y, \lambda} - \hat{\pi}_{S, \lambda}) \left[ 
    \frac{\hat{e} - e}{\hat{e}}(\hat{\mu}_1 - \mu_1) + 
    \frac{\hat{e} - e}{1 - \hat{e}}(\hat{\mu}_0 - \mu_0)\right] + 
    \tau_Y\left[ (\hat{\pi}_{Y, \lambda} - \pi_{Y, \lambda}) - 
    (\hat{\pi}_{S, \lambda} - \pi_{S, \lambda}) \right].\notag
    \end{align}
\end{proof}

\begin{lemma}[$\lambda$-gain lemma]
\label{lemma: A2}
    Let $\hat{\eta}$ be a collection of estimators for the nuisance parameters $\eta$. Then for any $\lambda \in (0,1]$,
    \begin{align}
        &E\left[\phi_{G, \lambda}(Y, A, X;\, \hat{\eta}) - \phi_{G, \lambda}(Y, A, X;\, \eta) \mid X\right] = \notag
        \\
        & \hat{\pi}_{S, \lambda} \left[ \frac{(\hat{e} - e)(\hat{\mu}_1 - \mu_1)}{\hat{e}} + 
    \frac{(\hat{e} - e)(\hat{\mu}_0 - \mu_0)}{1 - \hat{e}} \right] + 
    \tau_Y\left(\hat{\pi}_{S, \lambda} - \pi_{S, \lambda} \right).\notag
    \end{align}
\end{lemma}
\begin{proof}
    The results follow from the argument in the proof of Lemma \ref{lemma: A1}.
\end{proof}

\begin{lemma}[$\lambda$-efficiency lemma]
\label{lemma: A3}
    Let $\hat{\eta}$ be a collection of estimators for the nuisance parameters $\eta$. Then for any $\lambda \in (0,1]$,
    \begin{align}
        &E\left[\phi_{V, \lambda}(Y, A, X;\, \hat{\eta}) - \phi_{V, \lambda}(Y, A, X;\, \eta) \mid X \right] = \notag
        \\
        &  (\hat{\pi}_{S, \lambda} - \lambda)\left[ \frac{(\hat{e} - e)(\hat{\mu}_1 - \mu_1)}{\hat{e}} + 
    \frac{(\hat{e} - e)(\hat{\mu}_0 - \mu_0)}{1 - \hat{e}} \right] + 
    \tau_Y\left(\hat{\pi}_{S, \lambda} - \pi_{S, \lambda}\right).
    \notag
    \end{align}
\end{lemma}
\begin{proof}
    The results follows from the argument in the proof of Lemma \ref{lemma: A1}.
\end{proof}

\begin{lemma}[Lemma 2 in Kennedy et al. (2020)]
\label{lemma: B1}
Let \( \hat{f}(o) \) be a function estimated from training data \( \mathcal{O}_1 = \mathcal{D}_{1, 2} \cup \mathcal{D}_2 \), and let \( E_{I_{1,1}} \) be the empirical measure on \( \mathcal{O}_2 = \mathcal{D}_{1, 1} \) where \( \mathcal{O}_1 \) and \( \mathcal{O}_2 \) are iid samples from \( P \) with $\mathcal{O}_1 \perp\!\!\!\perp \mathcal{O}_2$. Write \( E_P(h) =  \int h(o) \, d P(o \mid \mathcal{O}_1) \) for the mean of any function \( h \) (possibly data-dependent) over a new observation. Then
\[
(E_{I_{1,1}} - E_P)(\hat{f} - f) = O_{P}\left(\frac{\|\hat{f} - f\|_2}{n^{1/2}}\right).
\]
\end{lemma}

\begin{proof}
Note that
\[
E\left[E_{I_{1,1}}(\hat{f} - f)\, |\, \mathcal{O}_1\right] = E(\hat{f}(O) - f(O) \mid \mathcal{O}_1) = E_P(\hat{f} - f),
\]
by identical distribution, \( \mathcal{O}_1 \perp\!\!\!\perp \mathcal{O}_2 \), and by definition of the operator \( E_P \). Moreover,
\[
\operatorname{Var}\left[(E_{I_{1,1}} - E_P)(\hat{f} - f) \mid \mathcal{O}_1\right] = \operatorname{Var}\left[E_{I_{1,1}}(\hat{f} - f) \mid \mathcal{O}_1\right] = \frac{1}{n}\operatorname{Var}\left(\hat{f} - f \mid \mathcal{O}_1\right) \leq \frac{1}{n}\|\hat{f} - f\|_2^2,
\]
by independence and identical distribution, and using the fact that \(\operatorname{Var}(W) \leq E(W^2)\) for any random variable \( W \). Thus,
\[
E\left\{\left[(E_{I_{1,1}} - E_P)(\hat{f} - f)\right]^2 \mid \mathcal{O}_1\right\} \leq 
\frac{1}{n} \|\hat{f} - f\|_2^2,
\]
and so
\[
P\left(\frac{n^{1/2}(E_{I_{1,1}} - E_P)(\hat{f} - f)}{\|\hat{f} - f\|_2} \geq t\right) = E\left[P\left(\frac{n^{1/2}(E_{I_{1,1}} - E_P)(\hat{f} - f)}{\|\hat{f} - f\|_2} \geq t \mid \mathcal{O}_1\right)\right] \leq \frac{1}{t^2},
\]
by applying Markov's inequality conditional on \( \mathcal{O}_1 \). For any \( \varepsilon > 0 \), we can choose \( t = \frac{1}{\varepsilon^{1/2}} \) to bound this probability by \( \varepsilon \), which yields the final result.
\end{proof}

\begin{lemma}
\label{lemma: B2}
    Let \( \hat{f} \) and \( f \) take any real values. Then
    \[
    |\mathbf{1}(\hat{f} > 0) - \mathbf{1}(f > 0)| \leq \mathbf{1}(|f| \leq |\hat{f} - f|).
    \]
\end{lemma}

\begin{proof}
    This follows since
    \[
    |\mathbf{1}(\hat{f} > 0) - \mathbf{1}(f > 0)| = 
    \mathbf{1}(\hat{f}, f \text{ have opposite sign}),
    \]
    and if \( \hat{f} \) and \( f \) have opposite signs, then
    \[
    |\hat{f}| + |f| = |\hat{f} - f|,
    \]
    which implies that \( |f| \leq |\hat{f} - f| \). Therefore, whenever \( |\mathbf{1}(\hat{f} > 0) - \mathbf{1}(f > 0)| = 1 \), it must also be the case that \( \mathbf{1}(|f| \leq |\hat{f} - f|) = 1 \), which yields the result.
\end{proof}

\begin{lemma}
\label{lemma: B3}
For any $\lambda \in (0, 1]$, we have the following results:
\begin{itemize}
    \item[(i)] Under Assumption 3(a) and 3(c), 
    \[
        \left\|\hat{\pi}_{Y, \lambda}(X) - \pi_{Y, \lambda}(X)\right\|_2 = o_{P}(1).
    \]
    \item[(ii)] Under Assumption 3(b) and 3(d), 
    \[
        \left\|\hat{\pi}_{S, \lambda}(X) - \pi_{S, \lambda}(X)\right\|_2 = o_{P}(1).
    \]
\end{itemize}
\end{lemma}

% \begin{lemma}
% \label{lemma: B3}
% Under the margin condition, for any $\lambda \in (0, 1]$, 
% \[
%     \left\|\hat{\pi}_{Y, \lambda}(X) - \pi_{Y, \lambda}(X)\right\|_2 = o_{P}(1),\ 
%     \left\|\hat{\pi}_{S, \lambda}(X) - \pi_{S, \lambda}(X)\right\|_2 = o_{P}(1).
% \]
% \end{lemma}

\begin{proof}
Observe that
\begin{align}
    \left\|\hat{\pi}_{Y, \lambda}(X) - \pi_{Y, \lambda}(X)\right\|_2^2 &= 
    E_P\left(\left[ \mathbf{1}(\hat{\tau}_Y > \hat{y}_{1 - \lambda})\mathbf{1}(\hat{\tau}_Y > 0) - 
    \mathbf{1}(\tau_Y > y_{1-\lambda})\mathbf{1}(\tau_Y > 0) \right]^2\right) \notag
    \\
    &= 
    E_P\left| \mathbf{1}(\hat{\tau}_Y > \hat{y}_{1 - \lambda})\mathbf{1}(\hat{\tau}_Y > 0) - 
    \mathbf{1}(\tau_Y > y_{1-\lambda})\mathbf{1}(\tau_Y > 0) \right| \notag
    \\
    &= E_P\left|
    \left[\mathbf{1}(\hat{\tau}_Y > \hat{y}_{1 - \lambda}) - 
    \mathbf{1}(\tau_Y > y_{1-\lambda})\right]
    \mathbf{1}(\hat{\tau}_Y > 0) + 
    \mathbf{1}(\tau_Y > y_{1-\lambda})
    \left[\mathbf{1}(\hat{\tau}_Y > 0) - 
    \mathbf{1}(\tau_Y > 0)\right]\right| \notag
    \\
    &\leq E_P\left|
    \mathbf{1}(\hat{\tau}_Y - \hat{y}_{1 - \lambda} > 0) - 
    \mathbf{1}(\tau_Y - y_{1-\lambda} > 0)\right| + 
    E_P\left|\mathbf{1}(\hat{\tau}_Y > 0) - 
    \mathbf{1}(\tau_Y > 0)\right|. \label{eq:lem6}
\end{align}
By applying Lemma \ref{lemma: B2} to both terms, we obtain:
\begin{align}
    \eqref{eq:lem6}
    &\leq P_X\left(|\tau_Y - y_{1-\lambda}| < |\hat{\tau}_Y - \tau_Y| + 
    |\hat{y}_{1-\lambda} - y_{1-\lambda}|\right) + 
    P_X(|\tau_Y| < |\hat{\tau}_Y - \tau_Y|) \notag
    \\
    &\leq P_X\left(|\tau_Y - y_{1-\lambda}| \leq t\right) + 
    P_X\left(|\hat{\tau}_Y - \tau_Y| + 
    |\hat{y}_{1-\lambda} - y_{1-\lambda}| > t\right) + P_X\left(|\tau_Y| \leq t\right) + 
    P_X\left(|\hat{\tau}_Y - \tau_Y| > t\right), \notag
\end{align}
Finally, by applying Assumption 3(a) and 3(c), and using Markov's inequality, we have:
\begin{align}
    \left\|\hat{\pi}_{Y, \lambda}(X) - \pi_{Y, \lambda}(X)\right\|_2^2 
    &\lesssim t^{\min\{\alpha_1, \beta_1\}} + 
    \frac{2E_P\left|\hat{\tau}_Y - \tau_Y\right| + 
    \left|\hat{y}_{1-\lambda} - y_{1-\lambda}\right|}{t} = o_{P}(1). \notag
\end{align}
The convergence of $\left\|\hat{\pi}_{S, \lambda}(X) - \pi_{S, \lambda}(X)\right\|_2$ can be established analogously by replacing $\tau_Y$ and $y_{1-\lambda}$ with $\tau_S$ and $s_{1-\lambda}$, respectively.
\end{proof}

\begin{corollary}
\label{cor: lambda_1_convergence}
When $\lambda = 1$, $\pi_{Y, \lambda}(X)$ and $\pi_{S, \lambda}(X)$ reduce to $\pi_{Y}(X)$ and $\pi_{S}(X)$, respectively. We have the following results:

\begin{itemize}
    \item[(i)] Under Assumption 3(a), 
    \[
        \left\|\hat{\pi}_{Y}(X) - \pi_{Y}(X)\right\|_2 = o_{P}(1).
    \]
    \item[(ii)] Under Assumption 3(b), 
    \[
        \left\|\hat{\pi}_{S}(X) - \pi_{S}(X)\right\|_2 = o_{P}(1).
    \]
\end{itemize}
\end{corollary}

\begin{proof}
We focus on the proof for (i), as (ii) is analogous.
For $\lambda = 1$, the estimation error simplifies to
\begin{align}
    \left\|\hat{\pi}_{Y}(X) - \pi_{Y}(X)\right\|_2^2 
    &= E_P\left| \mathbf{1}(\hat{\tau}_Y > 0) - \mathbf{1}(\tau_Y > 0) \right|. \notag
\end{align}
Proceeding analogously to the proof of Lemma \ref{lemma: B3}, and invoking Lemma \ref{lemma: B2} under Assumption 3(a), we have
\begin{align}
    \left\|\hat{\pi}_{Y}(X) - \pi_{Y}(X)\right\|_2^2 
    &\lesssim t^{\alpha_1} + \frac{E_P\left|\hat{\tau}_Y - \tau_Y\right|}{t} =o_P(1). \notag
\end{align}
\end{proof}

\begin{lemma}
\label{lemma: B4}
Assume that $\|\tau_Y(X)\|_{\infty} \leq M < \infty$. We have the following results:
\begin{itemize}
    \item[(i)] Under Assumption 3(a), 
    \begin{equation}
        E_P\left|\tau_Y(X)\left[\hat{\pi}_Y(X) - \pi_Y(X)\right]\right| \leq \| \tau_Y - \hat{\tau}_Y \|^{1 + \alpha_1}_{\infty}. \notag
    \end{equation}  
    \item[(ii)] Under Assumption 3(b), 
    \begin{equation}
        E_P\left|\tau_Y(X)\left[\hat{\pi}_S(X) - \pi_S(X)\right]\right| \lesssim \|\tau_S - \hat{\tau}_S\|_{\infty}^{\alpha_2}. \notag
    \end{equation}  
\end{itemize}

\end{lemma}

\begin{proof}
    For part (i), applying Lemma \ref{lemma: B2} yields
    \begin{align*}
       E_P\left|\tau_Y(X)\left[\hat{\pi}_Y(X) - \pi_Y(X)\right]\right| &=
        E_P\left| \tau_Y\left[\mathbf{1}(\tau_Y > 0) - 
        \mathbf{1}(\hat{\tau}_Y > 0)\right] \right|
        \\
        &\leq E_P \left|\tau_Y\, \mathbf{1}(|\tau_Y| 
        \leq |\tau_Y - \hat{\tau}_Y|)\right|
        \\
        &\leq \| \tau_Y - \hat{\tau}_Y \|^{1 + \alpha_1}_{\infty},
    \end{align*}
    where the last inequality follows from Assumption 3(a).
    
    For part (ii), Lemma \ref{lemma: B2} and the bound $\|\tau_Y(X)\|_{\infty} \leq M$ imply
    \begin{align*}
    E_P\left|\tau_Y\left[\hat{\pi}_S(X) - \pi_S(X)\right]\right| &\leq 
    M E_P\left| \mathbf{1}(\tau_S > 0) - 
    \mathbf{1}(\hat{\tau}_S > 0) \right|
    \\
    &\leq MP_X\left(|\tau_S| \leq |\tau_S - \hat{\tau}_S|\right)
    \\
    &\lesssim \|\tau_S - \hat{\tau}_S\|_{\infty}^{\alpha_2},
    \end{align*}
    where the final step is implied by Assumption 3(b).
\end{proof}

\begin{lemma}
\label{lemma: B5}
Assume that $\|\tau_Y(X)\|_{\infty} \leq M < \infty$. We have the following results:
\begin{itemize}
    \item[(i)] Under Assumptions 3(a) and 3(c),     
    $$ E_P\left|\tau_Y(X)\left[\hat{\pi}_{Y, \lambda}(X) - \pi_{Y, \lambda}(X)\right]\right| 
    \lesssim \| \tau_Y - \hat{\tau}_Y \|^{1 + \alpha_1}_{\infty} + 
    \left(\|\tau_Y - \hat{\tau}_Y\|_{\infty} + |y_{1-\lambda} - \hat{y}_{1-\lambda}|\right)^{\beta_1}. $$
    \item[(ii)] Under Assumptions 3(b) and 3(d), 
    $$ E_P\left|\tau_Y(X)\left[\hat{\pi}_{S, \lambda}(X) - \pi_{S, \lambda}(X)\right]\right| \lesssim
    \|\tau_S - \hat{\tau}_S\|_{\infty}^{\alpha_2}
    + \left(\|\tau_S - \hat{\tau}_S\|_{\infty} + |s_{1-\lambda} - \hat{s}_{1-\lambda}|\right)^{\beta_2}. $$
\end{itemize}
% if additionally $\|\hat{\tau}_Y(X) - \tau_Y(X)\|^{1+\alpha_1}_{\infty} = 
% o_{P}(n^{-1/2})$ and 
% $\|\hat{\tau}_S(X) - \tau_S(X)\|^{\alpha_2}_{\infty} = 
% o_{P}(n^{-1/2})$
% \[
%     E_P\left|\tau_Y(X)\left[\hat{\pi}_{Y, \lambda}(X) - \pi_{Y, \lambda}(X)\right]\right| = o_{P}(n^{-1/2}), 
%     E_P\left|\tau_Y(X)
%     \left[\hat{\pi}_{S, \lambda}(X) - \pi_{S, \lambda}(X)\right]\right| = 
%     o_{P}(n^{-1/2})
% \]    
\end{lemma}

\begin{proof}
    For result (i), by the triangle inequality, we decompose the term as follows:
    \begin{align}
        &E_P\left|\tau_Y(X)\left[\pi_{Y, \lambda}(X) - \hat{\pi}_{Y, \lambda}(X)\right]\right| \notag \\
        &= E_P\left| \tau_Y\left[\mathbf{1}(\tau_Y > 0)\mathbf{1}(\tau_Y > y_{1-\lambda}) - \mathbf{1}(\hat{\tau}_Y > 0)\mathbf{1}(\hat{\tau}_Y > \hat{y}_{1-\lambda}) \right]\right|\notag
        \\
        &\leq \underbrace{E_P\left|\tau_Y\,\mathbf{1}(\tau_Y > y_{1-\lambda}) \left[ \mathbf{1}(\tau_Y > 0) - \mathbf{1}(\hat{\tau}_Y > 0) \right]\right|}_{T_1} \notag
        \\ 
        &\quad + \underbrace{E_P\left|\tau_Y\,\mathbf{1}(\hat{\tau}_Y > 0) \left[ \mathbf{1}(\tau_Y > y_{1-\lambda}) - \mathbf{1}(\hat{\tau}_Y > \hat{y}_{1-\lambda}) \right]\right|}_{T_2}.\notag
    \end{align}
    For term $T_1$, by Lemma \ref{lemma: B4}, we have
    \begin{align}
        T_1 
        &\leq E_P\left|\tau_Y \left[ \mathbf{1}(\tau_Y > 0) - \mathbf{1}(\hat{\tau}_Y > 0) \right]\right| \notag \\
        &\leq \| \tau_Y - \hat{\tau}_Y \|^{1 + \alpha_1}_{\infty}.\notag
    \end{align}
    For term $T_2$, since $\|\tau_Y(X)\|_\infty \le M$, under Assumption 3(c), applying Lemma \ref{lemma: B2} yields
    \begin{align}
        T_2 &\lesssim E_P\left|\mathbf{1}(\tau_Y - y_{1-\lambda} > 0) - \mathbf{1}(\hat{\tau}_Y - \hat{y}_{1-\lambda} > 0) \right|\notag
        \\
        &\lesssim P_X\left(|\tau_Y - y_{1-\lambda}| \leq \left|\tau_Y - \hat{\tau}_Y - (y_{1-\lambda} - \hat{y}_{1-\lambda})\right|\right)
        \notag
        \\
        &\leq P_X\left(|\tau_Y - y_{1-\lambda}| \leq \left|\tau_Y - \hat{\tau}_Y\right| + \left|y_{1-\lambda} - \hat{y}_{1-\lambda}\right|\right)
        \notag
        \\
        &\lesssim \left(\|\tau_Y - \hat{\tau}_Y\|_{\infty} + |y_{1-\lambda} - \hat{y}_{1-\lambda}|\right)^{\beta_1}.
        \notag
        % \\
        % &\lesssim \|\tau_Y - \hat{\tau}_Y\|_{\infty}^{\beta_1} + |y_{1-\lambda} - \hat{y}_{1-\lambda}|^{\beta_1}.\notag
    \end{align}
    Combining these bounds, result (i) follows.

    For result (ii), similarly, we decompose the term for $S$:
    \begin{align}
        &E_P\left|\tau_Y(X)\left[\pi_{S, \lambda}(X) - \hat{\pi}_{S, \lambda}(X)\right]\right| \notag \\
        &\leq \underbrace{E_P\left|\tau_Y\mathbf{1}(\tau_S > s_{1-\lambda}) \left[ \mathbf{1}(\tau_S > 0) - \mathbf{1}(\hat{\tau}_S > 0) \right]\right|}_{T_3} \notag
        \\ 
        &\quad + \underbrace{E_P\left|\tau_Y\mathbf{1}(\hat{\tau}_S > 0) \left[ \mathbf{1}(\tau_S > s_{1-\lambda}) - \mathbf{1}(\hat{\tau}_S > \hat{s}_{1-\lambda}) \right]\right|}_{T_4}.\notag
    \end{align}
    For term $T_3$, given Lemma \ref{lemma: B4}, we have
    \begin{align}
        T_3 
        &\leq E_P\left|\tau_Y \left[ \mathbf{1}(\tau_S > 0) - \mathbf{1}(\hat{\tau}_S > 0) \right]\right| \notag \\
        &\lesssim \|\tau_S - \hat{\tau}_S\|_{\infty}^{\alpha_2}.\notag
    \end{align}
    For term $T_4$, by Assumption 3(d), we have
    \begin{align}
        T_4 &\lesssim 
        E_P\left|\mathbf{1}(\tau_S > s_{1-\lambda}) - 
        \mathbf{1}(\hat{\tau}_S > \hat{s}_{1-\lambda}) \right|\notag
        \\
        &\lesssim P_X\left(|\tau_S - s_{1-\lambda}| < \left|\tau_S - \hat{\tau}_S - (s_{1-\lambda} - \hat{s}_{1-\lambda})\right|\right)\notag
        \\
        &\lesssim \left(\|\tau_S - \hat{\tau}_S\|_{\infty} + |s_{1-\lambda} - \hat{s}_{1-\lambda}|\right)^{\beta_2}.\notag
    \end{align}
    Combining these results, result (ii) follows.
\end{proof}

\section{Proofs of Results in Section 5}
\label{appendix A.2}

In this section, we provide proofs for Theorems 1, 2, 3, and 4. 

\subsection{Proof of Theorem 1}
\label{pf:1}
\begin{proof}
Observe that 
\begin{align}
    \hat{R} - R &= 
    E_{I_{1,1}}\left[\phi_R(Y, A, X;\, \hat{\eta})\right] - E_P\left[\phi_R(Y, A, X;\, \eta)\right]\notag
    \\
    &= (E_{I_{1,1}}-E_P)\left[\phi_R(Y, A, X;\, \eta)\right] \nonumber\\
    &\quad + \underbrace{(E_{I_{1,1}}-E_P)\left[\phi_R(Y, A, X;\, \hat{\eta}) - \phi_R(Y, A, X;\, \eta)\right]}_{T_1} \nonumber
    \\
    &\quad + \underbrace{E_P\left[\phi_R(Y, A, X;\, \hat{\eta}) - \phi_R(Y, A, X;\, \eta)\right]}_{T_2}. \notag
\end{align}

Next, we analyze the term $T_1$. Given the sample splitting, $\hat{\eta}$ is independent of the data used in the empirical measure $E_{I_{1,1}}$. By Lemma \ref{lemma: B1}, to show $T_1 = o_{P}(n^{-1/2})$, it suffices to prove that the $L_2$ norm of the difference converges to zero in probability. The difference can be decomposed as:

\begin{align}
    &\left\| \phi_R(Y, A, X;\, \hat{\eta}) - \phi_R(Y, A, X;\, \eta) \right\|_2 \notag
    \\
    &= 
    \| (\hat{\pi}_Y - \hat{\pi}_S)
    \left\{ \frac{A}{\hat{e}} - \frac{1 - A}{1 - \hat{e}} \right\}(Y - \hat{\mu}_A) + 
    \hat{\tau}_Y(\hat{\pi}_Y - \hat{\pi}_S)\notag \\
    &\quad\quad - (\pi_Y - \pi_S)
    \left\{ \frac{A}{e} - \frac{1 - A}{1 - e} \right\}(Y - \mu_A) - \tau_Y(\pi_Y - \pi_S)\|_2. \label{eq:1}
\end{align}
By the triangle inequality, the Cauchy-Schwarz inequality, and the boundedness of $\tau_Y$, we have
\begin{align}
    \eqref{eq:1} &\leq \left\| (\hat{\pi}_Y - \hat{\pi}_S)\left[ \left\{ \frac{A}{\hat{e}} - \frac{1 - A}{1 - \hat{e}} \right\}(Y - \hat{\mu}_A) - 
    \left\{ \frac{A}{e} - \frac{1 - A}{1 - e} \right\}(Y - \mu_A) + \hat{\tau}_Y - \tau_Y \right] \right\|_2\notag
    \\&\quad\quad + \left\| \left[\hat{\pi}_Y - \hat{\pi}_S - (\pi_Y - \pi_S)\right]\left\{ \frac{A}{e} - \frac{1 - A}{1 - e} \right\} (Y - \mu_A) + 
    \tau_Y\left[\hat{\pi}_Y - \hat{\pi}_S - (\pi_Y - \pi_S)\right] \right\|_2\notag
    \\
    &\lesssim \|\hat{\mu}_0 - \mu_0\|_2 + \|\hat{\mu}_1 - \mu_1\|_2 + 
    \|\hat{e} - e\|_2 + 
    \|\hat{\pi}_Y - \pi_Y\|_2 + \|\hat{\pi}_S - \pi_S\|_2\notag
    \\&= o_{P}(1). \label{eq:4}
\end{align}
Combining Corollary \ref{cor: lambda_1_convergence} with the consistency of other nuisance parameters yields $\eqref{eq:4} = o_{P}(1)$, which leads to $T_1 = o_{P}(n^{-1/2})$.

It remains to analyze the bias term $T_2$. Using Lemma 2 and applying the law of total expectation yields
\begin{align}
    T_2& = E_P\left[\phi_R(Y, A, X;\, \hat{\eta}) - \phi_R(Y, A, X;\, \eta)\right]\notag
    \\
    &= E_P\left\{
    E_P\left[\phi_R(Y, A, X;\, \hat{\eta}) - \phi_R(Y, A, X;\, \eta)\mid X \right] 
    \right\}\notag
    \\
    & = E_P\left\{ (\hat{\pi}_Y - \hat{\pi}_S)\left[ \frac{(\hat{e} - e)(\hat{\mu}_1 - \mu_1)}{\hat{e}} + 
    \frac{(\hat{e} - e)(\hat{\mu}_0 - \mu_0)}{1 - \hat{e}} \right] + \tau_Y\left[\hat{\pi}_Y - \hat{\pi}_S - (\pi_Y - \pi_S)\right] \right\}\notag
    \\
    & \lesssim \left\|\hat{e} - e\right\|_2 \cdot 
    \left\|\hat{\mu}_1 - \mu_1\right\|_2+
    \left\|\hat{e} - e\right\|_2 \cdot \left\|\hat{\mu}_0 - \mu_0\right\|_2 +
    E_P\left|\tau_Y\,(\hat{\pi}_Y - \pi_Y)\right| + E_P\left|\tau_Y\,(\hat{\pi}_S - \pi_S)\right|\notag
    \\
    & \lesssim \left\|\hat{e} - e\right\|_2 \cdot 
    \left(\left\|\hat{\mu}_1 - \mu_1\right\|_2+ \left\|\hat{\mu}_0 - \mu_0\right\|_2\right) +
     \| \tau_Y - \hat{\tau}_Y \|^{1 + \alpha_1}_{\infty} + \|\tau_S - \hat{\tau}_S\|_{\infty}^{\alpha_2} \notag
    \\
    & = D_{1, n} + D_{2, n} + D_{3, n}, \notag
\end{align}
where the last step follows from Lemma \ref{lemma: B4}.

Since $\phi_R(Y, A, X;\, \eta)$ is a fixed, square-integrable function evaluated at the true parameter $\eta$, an application of the standard Central Limit Theorem (CLT) yields that $(E_{I_{1,1}}-E_P)\left[\phi_R(Y, A, X;\, \eta)\right] = O_P(n^{-1/2})$ and converges in distribution to a zero-mean normal random variable. 
\end{proof}

\begin{proof}[ of Proposition 1]
From Theorem 1, we have the decomposition 
\[\hat{R} - R = (E_{I_{1,1}} - E_P)[\phi_R(Y, A, X; \eta)] + T_2 + o_P(n^{-1/2}).\] 
Under the condition $D_{1,n} + D_{2,n} + D_{3,n} = o_P(n^{-1/2})$, it follows that \(n^{1/2}\left[T_2 + o_P(n^{-1/2})\right] = o_P(1).\)
As established at the end of the proof of Theorem 1, the leading empirical process term is asymptotically normal by the CLT. The desired result then follows immediately from Slutsky's theorem.
\end{proof}

\subsection{Proof of Theorem 2}
The proofs of Theorems 2, 3 and 4 follow the same structure as the proof of Theorem 1.
\begin{proof}
We adopt the same decomposition as in Theorem 1:
\begin{align}
    \hat{R}(\lambda) - R(\lambda) 
    &= E_{I_{1,1}}\left[\phi_{R, \lambda}(Y, A, X;\, \hat{\eta})\right] - E_P\left[\phi_{R, \lambda}(Y, A, X;\, \eta)\right]\notag
    \\
    &= (E_{I_{1,1}}-E_P)\left[\phi_{R, \lambda}(Y, A, X;\, \eta)\right] \notag\\
    &\quad + \underbrace{(E_{I_{1,1}}-E_P)\left[\phi_{R, \lambda}(Y, A, X;\, \hat{\eta}) - 
    \phi_{R, \lambda}(Y, A, X;\, \eta)\right]}_{T_1} \notag\\
    &\quad + \underbrace{E_P\left[\phi_{R, \lambda}(Y, A, X;\, \hat{\eta}) - 
    \phi_{R, \lambda}(Y, A, X;\, \eta)\right]}_{T_2}.\notag
\end{align}

Next, we analyze the term $T_1$. Analogous to the proof of Theorem 1, applying Lemma \ref{lemma: B1} requires bounding the $L_2$ norm difference:

\begin{align}
    &\left\| \phi_{R, \lambda}(Y, A, X;\, \hat{\eta}) - \phi_{R, \lambda}(Y, A, X;\, \eta) \right\|_2 \notag
    \\
    &= 
    \bigg\| (\hat{\pi}_{Y, \lambda} - 
    \hat{\pi}_{S, \lambda})
    \left\{ \frac{A}{\hat{e}} - 
    \frac{1 - A}{1 - \hat{e}} \right\}(Y - \hat{\mu}_A) + 
    \hat{\tau}_Y(\hat{\pi}_{Y, \lambda} - 
    \hat{\pi}_{S, \lambda}) \notag
    \\ 
    &\quad\quad - (\pi_{Y, \lambda} - \pi_{S, \lambda})
    \left\{ \frac{A}{e} - \frac{1 - A}{1 - e} \right\}(Y - \mu_A) - \tau_Y(\pi_{Y, \lambda} - \pi_{S, \lambda})\bigg\|_2\notag
    \\
    &\leq \left\| (\hat{\pi}_{Y, \lambda} - \hat{\pi}_{S, \lambda})\left[ \left\{ \frac{A}{\hat{e}} - \frac{1 - A}{1 - \hat{e}} \right\}(Y - \hat{\mu}_A) - 
    \left\{ \frac{A}{e} - \frac{1 - A}{1 - e} \right\}(Y - \mu_A) + \hat{\tau}_Y - \tau_Y \right] \right\|_2\notag
    \\&\quad\quad + \left\| 
    \left[\hat{\pi}_{Y, \lambda} -
    \hat{\pi}_{S, \lambda} - (\pi_{Y, \lambda} - \pi_{S, \lambda})\right]\left\{ \frac{A}{e} - \frac{1 - A}{1 - e} \right\} (Y - \mu_A) + 
    \tau_Y\left[\hat{\pi}_{Y, \lambda} - \hat{\pi}_{S, \lambda} - (\pi_{Y, \lambda} - \pi_{S, \lambda})\right] \right\|_2\notag
    \\
    &\lesssim \|\hat{\mu}_0 - \mu_0\|_2 + \|\hat{\mu}_1 - \mu_1\|_2 + 
    \|\hat{e} - e\|_2 + 
    \|\hat{\pi}_{Y, \lambda} - \pi_{Y, \lambda}\|_2 + 
    \|\hat{\pi}_{S, \lambda} - \pi_{S, \lambda}\|_2. \label{eq:2}
\end{align}
Combining Lemma \ref{lemma: B3} with the consistency of other nuisance parameters yields $\eqref{eq:2} = o_{P}(1)$, which leads to $T_1 = o_{P}(n^{-1/2})$.

Proceeding as in the proof of Theorem 1, Lemma \ref{lemma: A1} and the law of total expectation yield:
\begin{align}
    T_2 &= E_P\left\{
    E_P\left[\phi_{R, \lambda}(Y, A, X;\, \hat{\eta}) - \phi_{R, \lambda}(Y, A, X;\, \eta)\, |\, X \right] 
    \right\}\notag\\
    &= E_P\left\{ (\hat{\pi}_{Y, \lambda} - \hat{\pi}_{S, \lambda})\left[ \frac{(\hat{e} - e)(\hat{\mu}_1 - \mu_1)}{\hat{e}} + 
    \frac{(\hat{e} - e)(\hat{\mu}_0 - \mu_0)}{1 - \hat{e}} \right] + \tau_Y\left[\hat{\pi}_{Y, \lambda} - 
    \hat{\pi}_{S, \lambda} - 
    (\pi_{Y, \lambda} - \pi_{S, \lambda})\right] \right\}. \notag
\end{align}
By the boundedness of $\tau_Y$ and applying Lemma \ref{lemma: B5}, we have
\begin{align}
    T_2 &\lesssim \left\|\hat{e} - e\right\|_2 \cdot\left\|\hat{\mu}_1 - \mu_1\right\|_2+ \left\|\hat{e} - e\right\|_2 
    \cdot\left\|\hat{\mu}_0 - \mu_0\right\|_2 +
    E_P\left|\tau_Y\left(\hat{\pi}_{Y, \lambda} - 
    \pi_{Y, \lambda}\right)\right| + E_P\left|\tau_Y\left(\hat{\pi}_{S, \lambda} - \pi_{S, \lambda}\right)\right|\notag
    \\
    & \lesssim \left\|\hat{e} - e\right\|_2 \cdot 
    \left(\left\|\hat{\mu}_1 - \mu_1\right\|_2+ \left\|\hat{\mu}_0 - \mu_0\right\|_2\right) +
    \| \tau_Y - \hat{\tau}_Y \|^{1 + \alpha_1}_{\infty} + 
    \|\tau_S - \hat{\tau}_S\|_{\infty}^{\alpha_2} \notag
    \\
    &\quad
    + \left(\|\tau_Y - \hat{\tau}_Y\|_{\infty} + 
    |y_{1-\lambda} - \hat{y}_{1-\lambda}|\right)^{\beta_1}
    + \left(\|\tau_S - \hat{\tau}_S\|_{\infty} + 
    |s_{1-\lambda} - \hat{s}_{1-\lambda}|\right)^{\beta_2}
    \notag
    \\
    & = D_{1, n} + D_{2, n} + D_{3, n} + D_{4, n} + D_{5, n}. \notag
\end{align}
Similar to the proof of Theorem 1, since $\phi_{R, \lambda}(Y, A, X;\, \eta)$ is fixed and square-integrable, the standard CLT yields that $(E_{I_{1,1}}-E_P)\left[\phi_{R, \lambda}(Y, A, X;\, \eta)\right] = O_P(n^{-1/2})$ and its asymptotic normality.
\end{proof}

\begin{proof}[ of Proposition 2]
Proceeding as in the proof of Proposition 1, the result follows directly from the CLT applied to the leading term and an application of Slutsky's theorem, given that the remainder terms are \(o_P(n^{-1/2})\).
\end{proof}

\subsection{Proof of Theorem 3}
\begin{proof}
We adopt the same decomposition as in Theorems 1 and 2:
\begin{align}
    \hat{G}(\lambda) - G(\lambda) 
    &= E_{I_{1,1}}\left[\phi_{G, \lambda}(Y, A, X;\, \hat{\eta})\right] - E_P\left[\phi_{G, \lambda}(Y, A, X;\, \eta)\right]\notag
    \\
    &= (E_{I_{1,1}}-E_P)\left[\phi_{G, \lambda}(Y, A, X;\, \eta)\right] \notag\\
    &\quad + \underbrace{(E_{I_{1,1}}-E_P)\left[\phi_{G, \lambda}(Y, A, X;\, \hat{\eta}) - 
    \phi_{G, \lambda}(Y, A, X;\, \eta)\right]}_{T_1} \notag\\
    &\quad + \underbrace{E_P\left[\phi_{G, \lambda}(Y, A, X;\, \hat{\eta}) - 
    \phi_{G, \lambda}(Y, A, X;\, \eta)\right]}_{T_2}.\notag
\end{align}

Next, we analyze the term $T_1$. Analogous to the proof of Theorem 2, applying Lemma \ref{lemma: B1} requires bounding the $L_2$ norm difference:

\begin{align}
    &\left\| \phi_{G, \lambda}(Y, A, X;\, \hat{\eta}) - \phi_{G, \lambda}(Y, A, X;\, \eta) \right\|_2 \notag
    \\
    &= 
    \bigg\| \hat{\pi}_{S, \lambda}
    \left\{ \frac{A}{\hat{e}} - 
    \frac{1 - A}{1 - \hat{e}} \right\}(Y - \hat{\mu}_A) + 
    \hat{\tau}_Y \hat{\pi}_{S, \lambda} \notag
    \\ 
    &\quad\quad - \pi_{S, \lambda}
    \left\{ \frac{A}{e} - \frac{1 - A}{1 - e} \right\}(Y - \mu_A) - \tau_Y \pi_{S, \lambda}\bigg\|_2\notag
    \\
    &\leq \left\| \hat{\pi}_{S, \lambda}\left[ \left\{ \frac{A}{\hat{e}} - \frac{1 - A}{1 - \hat{e}} \right\}(Y - \hat{\mu}_A) - 
    \left\{ \frac{A}{e} - \frac{1 - A}{1 - e} \right\}(Y - \mu_A) + \hat{\tau}_Y - \tau_Y \right] \right\|_2\notag
    \\&\quad\quad + \left\| 
    (\hat{\pi}_{S, \lambda} - \pi_{S, \lambda})\left\{ \frac{A}{e} - \frac{1 - A}{1 - e} \right\} (Y - \mu_A) + 
    \tau_Y(\hat{\pi}_{S, \lambda} - \pi_{S, \lambda}) \right\|_2\notag
    \\
    &\lesssim \|\hat{\mu}_0 - \mu_0\|_2 + \|\hat{\mu}_1 - \mu_1\|_2 + 
    \|\hat{e} - e\|_2 + 
    \|\hat{\pi}_{S, \lambda} - \pi_{S, \lambda}\|_2
    \\ 
    &\lesssim \|\hat{\pi}_{S, \lambda} - \pi_{S, \lambda}\|_2 + 
    o_{P}(1). \label{eq:3}
\end{align}
Combining Lemma \ref{lemma: B3} with the consistency of other nuisance parameters yields $\eqref{eq:3} = o_{P}(1)$, which leads to $T_1 = o_{P}(n^{-1/2})$.

For the term $T_2$, proceeding as in the proofs of Theorems 1 and 2, we apply Lemma \ref{lemma: A2} for the bias decomposition, followed by Lemma \ref{lemma: B5} (ii) and the boundedness of $\tau_Y$:
\begin{align}
    T_2 &= E_P\left\{ \hat{\pi}_{S, \lambda}\left[ \frac{(\hat{e} - e)(\hat{\mu}_1 - \mu_1)}{\hat{e}} + 
    \frac{(\hat{e} - e)(\hat{\mu}_0 - \mu_0)}{1 - \hat{e}} \right] + \tau_Y(\hat{\pi}_{S, \lambda} - \pi_{S, \lambda}) \right\} \notag \\
    &\lesssim \left\|\hat{e} - e\right\|_2 \left(\left\|\hat{\mu}_1 - \mu_1\right\|_2+ 
    \left\|\hat{\mu}_0 - \mu_0\right\|_2\right) +
    E_P\left|\tau_Y\left(\hat{\pi}_{S, \lambda} - \pi_{S, \lambda}\right)\right|\notag
    \\
    & \lesssim \left\|\hat{e} - e\right\|_2 \left(\left\|\hat{\mu}_1 - \mu_1\right\|_2+ 
    \left\|\hat{\mu}_0 - \mu_0\right\|_2\right) + \|\tau_S - \hat{\tau}_S\|_{\infty}^{\alpha_2} + \left(\|\tau_S - \hat{\tau}_S\|_{\infty} + 
        |s_{1-\lambda} - \hat{s}_{1-\lambda}|\right)^{\beta_2}\notag
    \\
    & = D_{1, n} + D_{3, n} + D_{5, n}.\notag
\end{align}
By the same CLT argument as in Theorem 1, it follows that $(E_{I_{1,1}}-E_P)\left[\phi_{G, \lambda}(Y, A, X;\, \eta)\right] = O_P(n^{-1/2})$ and converges in distribution to a zero-mean normal random variable.
\end{proof}

\begin{proof}[ of Proposition 3]
The proof follows exactly the same arguments as those for Proposition 1.
\end{proof}

\subsection{Proof of Theorem 4}
\begin{proof}
We begin with the decomposition:
\begin{align}
    \hat{V}(\lambda) - V(\lambda) 
    &= (E_{I_{1,1}}-E_P)\left[\phi_{V, \lambda}(Y, A, X;\, \eta)\right] \notag\\
    &\quad + \underbrace{(E_{I_{1,1}}-E_P)\left[\phi_{V, \lambda}(Y, A, X;\, \hat{\eta}) - 
    \phi_{V, \lambda}(Y, A, X;\, \eta)\right]}_{T_1} \notag\\
    &\quad + \underbrace{E_P\left[\phi_{V, \lambda}(Y, A, X;\, \hat{\eta}) - 
    \phi_{V, \lambda}(Y, A, X;\, \eta)\right]}_{T_2}.\notag
\end{align}

First, to show $T_1 = o_{P}(n^{-1/2})$ via Lemma \ref{lemma: B1}, we bound the $L_2$ norm difference:
\begin{align}
    &\left\| \phi_{V, \lambda}(Y, A, X;\, \hat{\eta}) - \phi_{V, \lambda}(Y, A, X;\, \eta) \right\|_2 \notag
    \\
    &\leq \left\| (\hat{\pi}_{S, \lambda} - \lambda)\left[ \left\{ \frac{A}{\hat{e}} - \frac{1 - A}{1 - \hat{e}} \right\}(Y - \hat{\mu}_A) - 
    \left\{ \frac{A}{e} - \frac{1 - A}{1 - e} \right\}(Y - \mu_A) + \hat{\tau}_Y - \tau_Y \right] \right\|_2\notag
    \\&\quad\quad + \left\| 
    (\hat{\pi}_{S, \lambda} - \pi_{S, \lambda})
    \left\{ \frac{A}{e} - \frac{1 - A}{1 - e} \right\} (Y - \mu_A) + 
    \tau_Y(\hat{\pi}_{S, \lambda} - \pi_{S, \lambda}) \right\|_2\notag
    \\
    &\lesssim \|\hat{\mu}_0 - \mu_0\|_2 + \|\hat{\mu}_1 - \mu_1\|_2 + 
    \|\hat{e} - e\|_2 + 
    \|\hat{\pi}_{S, \lambda} - \pi_{S, \lambda}\|_2 \notag \\
    &= o_{P}(1). \notag
\end{align}

Next, we apply Lemma \ref{lemma: A3} to decompose the term $T_2$. Then, using Lemma \ref{lemma: B5} (ii) and the boundedness of $\tau_Y$, we obtain:
\begin{align}
    T_2 &= E_P\left\{ (\hat{\pi}_{S, \lambda} - \lambda)\left[ \frac{(\hat{e} - e)(\hat{\mu}_1 - \mu_1)}{\hat{e}} + 
    \frac{(\hat{e} - e)(\hat{\mu}_0 - \mu_0)}{1 - \hat{e}} \right] + 
    \tau_Y(\hat{\pi}_{S, \lambda} - \pi_{S, \lambda}) \right\}\notag
    \\
    &\lesssim \left\|\hat{e} - e\right\|_2 \left(\left\|\hat{\mu}_1 - \mu_1\right\|_2+ 
    \left\|\hat{\mu}_0 - \mu_0\right\|_2\right) +
    E_P\left|\tau_Y\left(\hat{\pi}_{S, \lambda} - \pi_{S, \lambda}\right)\right|
    \notag
    \\
    & \lesssim \left\|\hat{e} - e\right\|_2 \left(\left\|\hat{\mu}_1 - \mu_1\right\|_2+ 
    \left\|\hat{\mu}_0 - \mu_0\right\|_2\right) + \|\tau_S - \hat{\tau}_S\|_{\infty}^{\alpha_2} + \left(\|\tau_S - \hat{\tau}_S\|_{\infty} + 
        |s_{1-\lambda} - \hat{s}_{1-\lambda}|\right)^{\beta_2}\notag
    \\
    & = D_{1, n} + D_{3, n} + D_{5, n}.\notag
\end{align}
As in previous proofs, an application of the standard CLT yields $(E_{I_{1,1}}-E_P)\left[\phi_{V, \lambda}(Y, A, X;\, \eta)\right] = O_P(n^{-1/2})$, which is asymptotically normal.
\end{proof}

\begin{proof}[ of Proposition 4]
The result follows from an argument identical to that of Proposition 1.
\end{proof}

 % =============================================================

 %====================================================

\end{document}